\documentclass{IEEEtran}

\usepackage{stfloats}
\usepackage{amsfonts}
\usepackage{amssymb}
\usepackage{amsthm}
\usepackage{cite}
\usepackage{float}
\usepackage{color}
\usepackage{stfloats,fancyhdr}
\usepackage{amsmath,bm}
\usepackage{algorithm}
\usepackage{algorithmic}
\usepackage{multirow}
\usepackage{changepage}
\usepackage[normalem]{ulem}
\usepackage{amsthm}
\usepackage{balance}

\newtheorem{theorem}{Theorem}
\newtheorem{lemma}{Lemma}

\newtheorem{definition}{Definition}
\newtheorem{assumption}{Assumption}
\newtheorem{remark}{Remark}

\ifCLASSINFOpdf
\usepackage[pdftex]{graphicx}
\DeclareGraphicsExtensions{.pdf,.jpeg,.png}
\else
\usepackage[dvips]{graphicx}
\DeclareGraphicsExtensions{.eps}
\fi

\usepackage{subfigure}
\usepackage{fancybox,dashbox}
\usepackage{authblk}
\usepackage{tabularx}

\usepackage{mathtools}

\mathtoolsset{showonlyrefs}
\usepackage{url}

\usepackage{breakurl}
\usepackage[breaklinks]{hyperref}

\newcommand{\myexpect}[1]{\mathsf{E}\left[#1\right]}

\newcommand{\myprob}[1]{\mathsf{Prob}\left[#1\right]}

\newcommand\aug{\fboxsep=-\fboxrule\!\!\!\fbox{\strut}\!\!\!}
\usepackage[margin=0.3in]{geometry}

\begin{document}
		
\title{\huge Anytime Control under Practical Communication Model}

\author{Wanchun Liu, Daniel E.\ Quevedo, Yonghui Li,  
%	Karl Henrik Johansson 
	and Branka Vucetic
	\vspace{-0.9cm}
}

\maketitle
\begin{abstract}
\let\thefootnote\relax\footnote{W. Liu, Y. Li and B. Vucetic are with School of Electrical and Information Engineering, The University of Sydney, Australia.
	Emails:	\{wanchun.liu,\ yonghui.li,\ branka.vucetic\}@sydney.edu.au. 
D. E. Quevedo is with the School of Electrical Engineering and Robotics, Queensland University of Technology (QUT), Brisbane, Australia.	Email: daniel.quevedo@qut.edu.au.
%K. H. Johansson is with School of Electrical Engineering and Computer Science, KTH Royal
%Institute of Technology, Stockholm, Sweden. Email: kallej@kth.se.	
}We investigate a novel anytime control algorithm for wireless networked control with random dropouts. The controller computes sequences of tentative future control commands using time-varying (Markovian)  computational resources. The sensor-controller and controller-actuator channel states are spatial- and time-correlated, and are modeled as a multi-state Markov process.
To compensate the effect of packet dropouts, a dual-buffer mechanism is proposed. 
We develop a novel cycle-cost-based approach to obtain the stability conditions on the nonlinear plant, controller, network and computational resources.
\end{abstract}

\begin{IEEEkeywords}
	Control over communications, nonlinear systems, stability of nonlinear systems, Markov fading channels.
\end{IEEEkeywords}

%\begin{IEEEkeywords}
%Estimation, Kalman filtering, linear systems, stability, mean-square error, Markov fading channel
%\end{IEEEkeywords}
%\vspace{-0.3cm}
\section{Introduction}
During   past decades, significant attention has focused on embedded or networked control systems that have limited and time-varying controller's computation capability due to high requirements on multitasking operations.
In particular, assuming constant and limited computational resources, bounds on computational time of specific optimization algorithms for achieving stability were derived in~\cite{McGovern,McGovern1}.
For time-varying computational resources, a dynamic computation task scheduling method was proposed for model predictive controllers~\cite{Henriksson}.
%, where  control inputs were not updated   during the computational delay of for optimization algorithms of MPC.
On-demand computation scheduling of control input based on plant states were investigated for periodic, event-triggered and self-triggered policies in~\cite{ondemand1,ondemand2,ondemand3}, respectively.

Another stream of research considers anytime algorithms for robust control and making efficient use of time-varying computational resources.
In general, an anytime algorithm can provide \emph{a} solution even with limited computational resources, and refines the solution when more resources are available.
In the pioneering work~\cite{bhattacharya2004anytime}, an anytime control system was proposed, where the number of updated states varies with the available computation time known to the controller \emph{a priori}.
In~\cite{GuptaAnytime}, an anytime control algorithm for a multi-input linear system was proposed for the scenario when the computation availability is \emph{unknown a priori}. The main idea was to first calculate the most important component of the control vector and then calculate the less important ones as more computational resource becomes available.
In~\cite{anytime5}, a  sequence-based anytime control method was proposed, which can calculate a tentative sequence of future control input for as many time steps as allowed by the available computational resources at each time step. The pre-calculated control sequence can compensate for the time steps when no computational resource is available for control calculations.
Following this work, sequence-based anytime control systems with Markovian processor availability, event-triggered sensor updates and multiple control laws were investigated in~\cite{anytime2,anytime3,anytime4}, respectively.
%In these anytime-control related works, either both sensor-controller
In \cite{bhattacharya2004anytime,GuptaAnytime,anytime5,anytime2}, the sensor-controller and controller-actuator channels were assumed to be perfect and error-free.
In~\cite{anytime3} and~\cite{anytime4}, the sensor-controller channel was assumed to have independent and identically distributed (i.i.d.) packet dropouts\footnote{Note that i.i.d. packet dropout channel is very commonly considered in the literature of networked control~\cite{ZHAN2015214}.}, i.e.,   only binary-level (on-off) channel states, were considered, while the controller-sensor channel was assumed to be perfect.

{\color{black}\emph{Motivation.}
The existing work of anytime control~\cite{bhattacharya2004anytime,GuptaAnytime,anytime5,anytime2,anytime3,anytime4} cannot effectively handle the fully distributed networked control scenario, where both the sensor-controller and controller-actuator communication links are wireless.
Moreover, the existing research only considered simple wireless channel models, which cannot capture the key features of practical wireless channels that are time-varying and correlated~\cite{Parastoo}.
Therefore, anytime control design and analysis for fully distributed networked control system in practical wireless channels are critical in practice, but also present new challenges.}

{\color{black}\emph{Novelty and contributions.}
In this work, we consider for the first time the sequence-based anytime control of a wireless networked control system (WNCS) in a generalized \textbf{dual imperfect channel}, where the sensor-controller and controller-actuator channel states are \textbf{spatial- and time-correlated} and are modeled as a \emph{multi-state  Markov process}.
Different from most of the existing works~\cite{anytime5,anytime2,anytime3,anytime4}, where only the controller has a buffer to keep the calculated sequence of control inputs, we propose to use anytime control with \textbf{buffers at both the controller and the actuator nodes}. The latter is used to compensate for dropouts in the   controller-actuator channel.
Moreover, the available computational resource of the controller is allowed to be time-correlated, modeled as a multi-state Markov process.
Such a \textbf{dual-channel-dual-buffer} anytime control system has practical advantages but brings significant challenges to its  analysis due to the complex system state updating rule, compared to previous setups.
We propose a novel cycle-cost-based approach to
derive sufficient conditions for stochastic  stability of the overall WNCS. Our stability conditions are stated in terms of plant dynamics, network dynamics,  buffer properties and computational resource dynamics.  We further show that, under suitable assumptions, the conditions derived guarantee robust stability when plant disturbances are taken into account.}

The remainder of the paper is organized as follows:  
Section~II presents the system model of anytime control in the dual-channel-dual-buffer WNCS.
Section~III develops the stability condition.
Section~IV provides robust stability analysis.
Section~V draws conclusions.

\emph{Notation:} Sets are denoted by calligraphic capital letters, e.g., $\mathcal{A}$.
$\mathcal{A} \backslash \mathcal{B}$ denotes set subtraction.
Matrices and vectors are denoted by capital and lowercase upright bold letters, e.g., $\mathbf{A}$ and $\mathbf{a}$, respectively.
%$\vert \mathcal{A}\vert$ denotes the cardinality of the set $\mathcal{A}$.
$\mathsf{E}\left[A\right]$ is the expectation of the random variable $A$.
%The conditional expectation $\mathsf{E}\left[A \vert B\right] = 0$ if 
The conditional probability $\myprob{{A}\vert {B}}=0$ if $\myprob{{B}}=0$.
$(\cdot)^{\top}$ is the matrix transpose operator.
{\color{black}$| \mathbf{v} |$ is the Euclidean norm of vector $\mathbf{v}$. }
%$\| \mathbf{v} \|_1$ is the sum of the vector $\mathbf{v}$'s elements. 
%$|\mathbf{v}| \triangleq \sqrt{\mathbf{v}^\top \mathbf{v}}$ is the Euclidean norm of a vector $\mathbf{v}$.
%$\text{Tr}(\cdot)$ is the trace operator. $\text{diag}\{v_1,v_2,...,v_K\}$ denotes the diagonal matrix with the diagonal elements $\{v_1,v_2,...,v_K\}$. 
$\mathbb{N}$ and $\mathbb{N}_0$ denote the sets of positive and non-negative integers, respectively.
$\mathbb{R}^m$ denotes the $m$-dimensional Euclidean space.
%$\rho(\mathbf{A})$ is the spectral radius of $\mathbf{A}$, i.e., the  largest absolute value of its eigenvalues.
%$\left[u\right]_{B \times B}$ denotes the $B \times B$ matrix with identical elements~$u$.
$[\mathbf{A}]_{j,k}$ and $[\mathbf{v}]_{j}$ denote the element at the $j$th row and $k$th column of a matrix $\mathbf{A}$, and the $j$th element of a vector $\mathbf{v}$, respectively.
$\lambda_{\max}(\mathbf{A})$ denotes the spectral radius of $\mathbf{A}$.
$\text{diag}\{\mathbf{v}\}$ denotes the diagonal matrix generated by the vector $\mathbf{v}$.
$\{v\}_{\mathbb{N}_0}$ denotes the semi-infinite sequence $\{v_0,v_1,\cdots\}$.
A function $\phi: \mathbb{R}_{\geq 0} \rightarrow \mathbb{R}_{\geq 0}$ is of class-$\mathcal{K}_{\infty}$ ($\phi \in \mathcal{K}_{\infty}$) if it is continuous, strictly increasing, and zero at zero.
$\mathbf{0}_{i}$ and $\mathbf{0}_{i \times j}$  denotes the all-zero $i \times i$ and $i \times j$ matrices, respectively. $\mathbf{A}=\mathbf{0}$ indicates that $\mathbf{A}$ has all zero elements.

\section{Anytime Control in a Dual-Channel-Dual-Buffer WNCS} \label{sec:sys}
We consider a WNCS consisting of a plant system, a remote controller and a wireless network placed between the plant and the controller.
A sensor sends its measurements of the plant to the controller, and the controller computes and sends control commands to a remote actuator via the wireless network as illustrated in Fig.~\ref{fig:system_model}. The dual wireless channel (i.e., the sensor-controller and controller-actuator channels) setup is different from \cite{anytime2,anytime3,anytime5}, which assumed either perfect transmissions in two channels or a single (imperfect) wireless channel from the sensor to the controller.
\begin{figure}[t]
	\centering\includegraphics[scale=0.8]{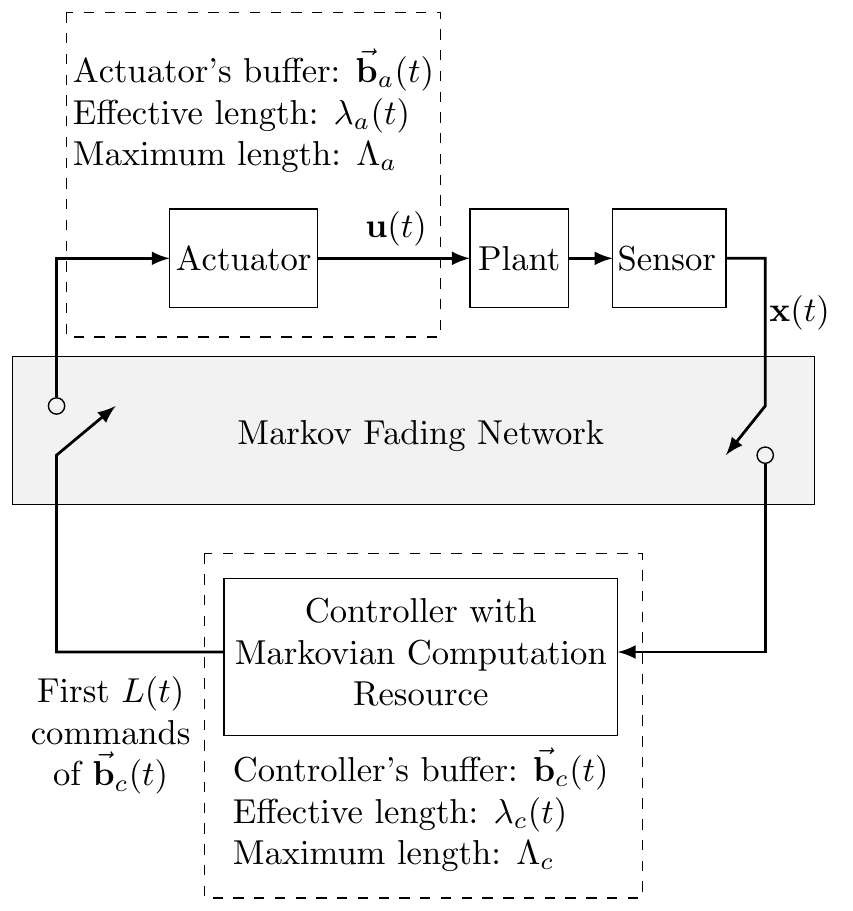}
	\vspace{-0.3cm}
	\caption{The dual-channel-dual-buffer WNCS.}
	\label{fig:system_model}
	\vspace{-0.5cm}	
\end{figure}
Each sampling period of the plant, denoted by $T_s$, is divided into four phases: sensor-controller (S-C) transmission, command computation, controller-actuator (C-A) transmission and implementation of control as illustrated in Fig.~\ref{fig:process}.
\begin{figure*}[t]
	\centering\includegraphics[scale=0.8]{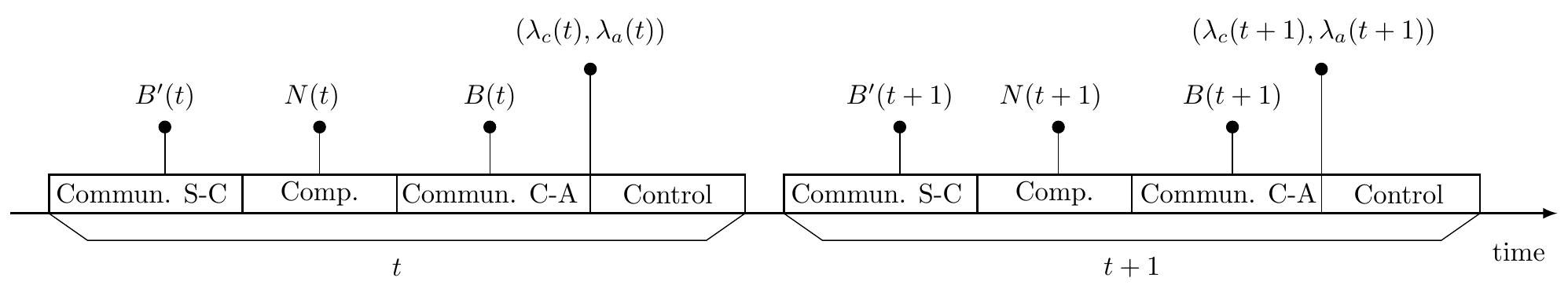}
	\vspace{-0.5cm}	
	\caption{A communications-computing-control process.}
	\label{fig:process}
	\vspace{-0.5cm}	
\end{figure*}

We consider a process-noise-free discrete-time non-linear plant model
\begin{equation}\label{eq:sys}
\mathbf{x}(t+1) = f(\mathbf{x}(t),\mathbf{u}(t)),
\end{equation}
where $\mathbf{x}(t)\in \mathbb{R}^{l_s}$ and $\mathbf{u}(t)\in \mathbb{R}^{l_u}$ are the plant state and the control input at time $t$.
Note that a more practical model with process noise will be investigated in Section~\ref{sec:robust}.

As a consequence of stochastic computational resources at the controller and   packet dropouts in the sensing and control channels (which will be described later in more detail), the plant may have to operate  in open loop for arbitrarily long time intervals. This may lead to performance degradation and potential loss of stability. Thus, throughout this work we will analyze the stability conditions of the considered system where the stability is defined as follows.

\begin{definition}\label{def:stability}
	\normalfont	
	The process-noise-free dynamical system \eqref{eq:sys} is stochastically stable, if for some $\psi \in \mathcal{K}_\infty$, the expected value $\sum_{k=0}^{\infty} \myexpect{\psi(|\mathbf{x}(k)|)}<\infty$.
\end{definition}

Our standing assumption is that the plant is globally controllable (in the idealized closed loop case):
\begin{assumption}[\!\!\cite{anytime2,anytime3,anytime5}] \label{assum:rho}
\normalfont
There exist functions $V: \mathbb{R}^{l_s} \rightarrow \mathbb{R}_{\geq 0}, \psi_1,\psi_2 \in \mathcal{K}_\infty$, a constant $\rho\in(0,1)$, and a control policy $\kappa: \mathbb{R}^{l_s} \rightarrow \mathbb{R}^{l_u}$, such that for all $x\in\mathbb{R}^{l_s}$
\begin{equation}\notag
\begin{aligned}
&\psi_1(|\mathbf{x}|) \leq V(\mathbf{x}) \leq \psi_2(|\mathbf{x}|)\\
&V(f(\mathbf{x},\kappa(\mathbf{x}))) \leq \rho V(\mathbf{x}).
\end{aligned}
\end{equation}
\end{assumption}

\begin{assumption}[\!\!\cite{anytime2,anytime3,anytime5}]\label{assum:alpha}
	\normalfont	
	There exists $\alpha>0$ such that $$V(f(\mathbf{x},\mathbf{0}))< \alpha V(\mathbf{x}), \forall \mathbf{x} \in \mathbb{R}^{l_s}.$$
	The initial plant state satisfies  $$\myexpect{\psi_2( \vert \mathbf{x}(0)\vert )}<\infty,$$ where $\psi_2 \in \mathcal{K}_\infty$ as in Assumption~\ref{assum:rho}.
\end{assumption}

\subsection{Dual Markov Fading Channels}
We consider wireless fading channels for the S-C and C-A transmissions~\cite{KangJIoT,Kang2019ICC,KangTWC}, where wireless channel status varies with time due to multi-path propagation and shadowing caused by obstacles affecting the radio-frequency (RF) wave propagation.
The time-varying channel conditions can be modeled as Markov processes~\cite{Parastoo}.
Furthermore, in practice, the two channel conditions can be correlated, named as spatial correlation, which is caused by the same environment obstacles~\cite{tse2005fundamentals}.

For wireless packet transmissions, there exists a fundamental tradeoff between reliability, i.e., the packet drop probability, and data rate, which determines the amount of information bits that a packet can carry~\cite{Polyanskiy}. For a fixed channel condition, increasing the data rate of a packet can lead to a higher packet drop probability.
For a good channel condition, one can reduce the packet drop probability while maintain the fixed data rate. Alternatively, one can increase the data rate while maintaining the packet drop probability~\cite{Goldsmith}.
In this work, we will keep the data rate fixed for the S-C transmissions, since there is nothing gained from aggregating past sensor measurements in the state feedback case. However, it is well known that sending control sequences can be beneficial to compensate for packet dropouts. To accommodate this in a fading channel environment, we allow C-A transmissions to contain packets of varying data rate. The rates  depend on the channel condition and provide a guaranteed packet-dropout probability. Thus, a longer control sequence can be transmitted to the actuator under a better channel condition with the same reliability. 

%For a typical predictive control system, the controller with some computational resources can generate and send a sequence of predictive control commands to the actuator,
%while the sensor without any computational resource can only take measurement and send it to the controller at each time.
%Due to the different data transmission requirements in the C-A and the S-C channels, we adopt different Markov channel models for them.

\textbf{C-A channel.} At time slot $t$, the controller can at most transmit $B(t) \in \mathcal{B}\triangleq\{0,1,\cdots, \bar{B}\}$ commands to the actuator with a guaranteed packet drop probability $\bar{\gamma}$.
In other words, $B(t)$ denotes the C-A channel quality.
In this sense, $B(t)$ can be treated as the capacity of the channel under the packet drop probability requirement $\bar{\gamma}$. 
Let $\gamma(t)=1$ and $\gamma(t)=0$ denote the successful and failed transmissions in time slot $t$.

\textbf{S-C channel.} At time slot $t$, 
the S-C channel power gain takes values from $\{h_1,h_2,\cdots,h_{\bar{B}'}\}$. 
Let $B'(t) \in \mathcal{B}' \triangleq \{1,2,\cdots,\bar{B}'\}$ denote the index of the channel power gain. Thus, $B'(t)$ denotes the S-C channel quality. Let $\gamma'(t)=1$ and $\gamma'(t)=0$ denote the successful and failed transmissions in time slot $t$.
The packet drop probability at time~$t$ is 
\begin{equation}
\bar{\gamma}'(t)=g(h_{B'(t)}) \in  \{g(h_1),g(h_2),\cdots,g(h_{\bar{B}'})\},
\end{equation}
where $g(\cdot)$ is the packet drop probability function in terms of the channel power gain. 

Then, we assume that the joint C-A and S-C channel condition $\{(B,B')\}_\mathbb{N}$ is a time-homogeneous Markov process, and the state transition probability is given as
\begin{equation}\label{eq:prob_B}
\begin{aligned}
p_{i,j} &\triangleq \myprob{(B(t+1),B'(t+1))=b_j \vert (B(t),B'(t))=b_i}, \\
& \hspace{5.8cm}\forall  b_i,b_j\in \mathcal{B} \times \mathcal{B}'.
\end{aligned}
\end{equation}

\begin{remark}
	Our current channel model jointly considers both the spatial-correlated S-C and C-A channels, the time-correlated fading channel conditions and variable data rate requirements. To the best of our knowledge, this has never been considered in the literature of WNCSs and is more general than  existing models.
	For example, independent dual channels with i.i.d.\ packet dropouts were considered in~\cite{schenato2007foundations}, which is a special case of our model when the fading channels degrade to static ones and the channels' spacial correlation is  perfectly canceled.
\end{remark}

\subsection{Anytime Control with Dual Buffers}
When considering perfect transmission between the controller and the actuator, as in \cite{anytime2,anytime3,anytime5}, the system only needs one buffer at the controller to store the computed control commands.
If imperfect transmissions are taken into account, it is convenient to include a command buffer at the actuator to provide robustness against packet dropouts, see e.g.~\cite{QUEVEDO20121803} for a general packetized predictive control method\footnote{Markovian communication and computational resources were not considered in~\cite{QUEVEDO20121803}.}. 
Clearly, the dual-buffer system introduces a more complex state updating process.

Let $\Lambda_c$ and $\Lambda_a$ denote the length of the controller's and the actuator's buffer, respectively.
Then, the buffer state at the controller after its transmission phase is denoted as
\begin{equation}
\vec{\mathbf{b}}_c(t) \triangleq [\mathbf{b}_{c,1}(t)^\top,\mathbf{b}_{c,2}(t)^\top,\cdots, \mathbf{b}_{c,\Lambda_c}(t)^\top]^\top,
\end{equation}
where $\mathbf{b}_{c,i}(t)\in\mathbb{R}^{l_u},\forall i\in \{1,2,\cdots,\Lambda_c\}$.
The buffer state at the actuator right after the C-A transmission but before the implementation of a control command, i.e., the pre-control buffer, is denoted as
\begin{equation}
\vec{\mathbf{b}}_a(t)\triangleq [\mathbf{b}_{a,1}(t)^\top,\mathbf{b}_{a,2}(t)^\top,\cdots, \mathbf{b}_{a,\Lambda_a}(t)^\top]^\top,
\end{equation}
where $\mathbf{b}_{a,i}(t)\in\mathbb{R}^{l_u},\forall i\in \{1,2,\cdots,\Lambda_a\}$.
In general, the buffers $\vec{\mathbf{b}}_a(t)$ and $\vec{\mathbf{b}}_c(t)$ keep the calculated sequences of control command, and the buffer updating rules will be given in the following part.

The control input is the first element in the actuator's buffer, i.e.,
\begin{equation}
\mathbf{u}(t) = \mathbf{b}_{a,1}(t),
\end{equation}
which can be treated as the previously predicted control command for the current time slot.
The buffer state at the actuator right after the control implementation, i.e., the post-control state, is
\begin{equation}
\vec{\mathbf{b}}'_a(t) \triangleq \mathbf{S}_a \vec{\mathbf{b}}_a(t),
\end{equation}
where the buffer shift matrices are defined as
\begin{equation}
\color{black}
\mathbf{S}_i \triangleq \begin{bmatrix}
\mathbf{0}_{l_u} & \mathbf{I}_{l_u} & \mathbf{0}_{l_u} & \cdots & \mathbf{0}_{l_u}\\
\vdots & \ddots & \ddots & \ddots & \vdots \\
\mathbf{0}_{l_u} & \cdots & \mathbf{0}_{l_u} & \mathbf{I}_{l_u} & \mathbf{0}_{l_u}\\
\mathbf{0}_{l_u} & \cdots & \cdots  & \mathbf{0}_{l_u} & \mathbf{I}_{l_u}\\
\mathbf{0}_{l_u} & \cdots & \cdots & \cdots & \mathbf{0}_{l_u}\\
\end{bmatrix} \in \mathbb{R}^{\Lambda_i l_u \times \Lambda_i l_u}, i=c \text{ or } a.
\end{equation}

Let $N(t)$ denote the number of calculated tentative future control commands at time $t$. The process $\{N\}_{\mathbb{N}}$ is a time-homogeneous Markov process with the transition probability
\begin{equation} \label{eq:prob_N}
q_{i,j} \triangleq \myprob{N(t+1)=j \vert N(t)=i}, i,j\in \mathcal{N},
\end{equation}
where $\mathcal{N}\triangleq \{0,1,\cdots,\bar{N}\}$. It is assumed that $\bar{N} \leq \Lambda_c$.

The \textbf{controller}'s operations are described as
\begin{enumerate}
	\item  If $\gamma'(t)N(t)>0$, the controller has a new update from the sensor and is available for computation.	
	In this case, it discards all the existing commands in its buffer and generates a sequence of $N(t)$ control commands to control the plant in time slots $t$ to $(t+N(t)-1)$. 
	The sequence of tentative controls is
	\begin{equation}
	\vec{\mathbf{u}}(t) = [\mathbf{u}_1(t)^\top,\mathbf{u}_2(t)^\top,\cdots,\mathbf{u}_{N(t)}(t)^\top]^\top.
	\end{equation}
	The buffer state before transmission is written as
	$$[\vec{\mathbf{u}}(t)^\top,\underbrace{\mathbf{0}_{l_u \times 1}^\top,\cdots,\mathbf{0}_{l_u \times 1}^\top}_{\Lambda_c-N(t)}]^\top. $$
	
	Specifically, the controller calculates the control sequence based on the anytime control algorithm proposed in~\cite{anytime5}, which is rewritten as
	\begin{equation}
	\begin{aligned}
	&\mathbf{u}_i(t) = \kappa(\mathbf{x}_i'(t)),  i=1,\cdots,N(t)\\
	&\mathbf{x}'_i(t) = \begin{cases}
	\mathbf{x}(t), &i=1\\
	f(\mathbf{x}'_{i-1}(t),\mathbf{u}_{i-1}(t)),  &i=2,\cdots,N(t)
	\end{cases}\\
	\end{aligned}
	\end{equation}
	where $\kappa(\cdot)$ is defined in Assumption~\ref{assum:rho}.

	Considering the C-A channel capacity and the actuator's buffer length, the controller transmits $\min\{B(t),N(t),\Lambda_a\}$ commands to the actuator. If the transmission is successful, the buffer shifts by $\min\{B(t),N(t),\Lambda_a\}$ steps. Otherwise, the controller erases its buffer. This is because the first computed control command $\mathbf{u}_1(t)$ cannot be implemented in the current time slot $t$, and the rest of computed control commands, which are calculated based on the successful implementation of the first control command, become useless. It is clear that at time instances where   the actuator has run out of buffer contents,  we have   $\vec{\mathbf{b}}_a(t) = \mathbf{0}$ and $\mathbf{u}(t) = \mathbf{0} \neq \mathbf{u}_1(t)$. For the case that $\vec{\mathbf{b}}_a(t) \neq \mathbf{0}$, the predicted control input $\mathbf{u}(t)$ is equal to currently calculated control command $\mathbf{u}_1(t)$ only in the perfect process-noise-free scenario, and $\mathbf{u}(t)\neq\mathbf{u}_1(t)$ in general.	

\item If $\gamma'(t)=0$, the controller does not have a new update from the sensor. In this case, it does not generate any new control command. 
	If $N(t)=0$, the controller does not have the computational resource to generate any new control command.\\
	In these two cases, if the plant is out of control in the previous time slot, i.e., $\vec{\mathbf{b}}_a(t-1) = \mathbf{0}$, the controller erases its buffer due to the same reason in case 1); otherwise, the controller sends the buffered commands to the actuator  as much as it can, subject to the constraints of the C-A channel capacity and the actuator's buffer length.
\end{enumerate}

The \textbf{actuator}'s operations are described as
\begin{enumerate}
%	\item After the implementation of control in each time slot, the actuator's buffer is shifted by one step, i.e., $\mathbf{S}_a \vec{\mathbf{b}}_a(t-1)$, since the first command in the buffer is used for control. This applies to all the following scenarios.
	\item  If $\gamma(t)=0$, the actuator's buffer is shifted by one step, i.e., $\vec{\mathbf{b}}_a(t) = \mathbf{S}_a \vec{\mathbf{b}}_a(t-1)$, since the first command in the buffer of the previous time slot was used for control.
	\item  If $\gamma(t)=1$ and $\gamma'(t)N(t)>0$, the actuator erases the previous commands and stores the received ones.
{\color{black}This is because the newly calculated control commands are expected to perform better than the previous calculated commands, especially when the process noise has a large variance.}
\footnote{\color{black}	Note that when the number of tentative future commands in the newly received packet is less than the previous ones, we still need to erase all the previous commands, not just part of them. The reason is that each buffered control command at the actuator was calculated assuming the successful implementation of the previously calculated control command. Thus, if the first few control commands at the actuator's buffer are removed (i.e., cannot be applied for control), then the rest of the buffer is of little use. }

	\item  If $\vec{\mathbf{b}}_a(t-1) = \mathbf{0}$ and $\gamma'(t)N(t)=0$, no operation on the actuator's buffer is required as there is no new commands transmitted.
	\item  If $\vec{\mathbf{b}}_a(t-1) \neq \mathbf{0}$, $\gamma'(t)=1$ and $\gamma'(t)N(t)=0$, the actuator shifts its buffer by one step due to the same reason in 1) and stores the received commands in the buffer right after the existing commands.
\end{enumerate}

Let $\lambda_c(t)$ and $\lambda_a(t)$ denote the effective buffer lengths at the controller and the actuator, respectively. 
Intuitively, $\lambda_c(t)$ and $\lambda_a(t)$ jointly determine the closed-loop performance of plant, as a larger $\lambda_c(t)$ and a larger $\lambda_a(t)$ indicate that the plant will be delivered tentative control values  for a longer time.
Based on the controller's and the actuator's operations, the actuator is not necessary to have a larger buffer than the controller. Thus, we assume that $\Lambda_a \leq \Lambda_c$.

Let $L(t)$ denote the number of tentative commands to be transmitted. Based on the controller's operations, $L(t)$ can be written as
\begin{equation} \label{eq:V}
L(t) \!=\! \begin{cases}
\!\min\{B(t),N(t),\Lambda_a\}, &\text{\!\!\!\!\!\!if } \gamma'(t)N(t)>0\\
\!\min\{B(t),\lambda_c(t-1),\Lambda_a-\lambda_a(t-1)\}, & \begin{aligned}
&\text{\!\!\!\!\!if } \gamma'(t)N(t)=0, \\
&\!\!\!\!\!\lambda_a(t-1)\neq 0
\end{aligned} \\
\!0, &\text{\!\!\!\!\!\!otherwise.}\\
\end{cases}
\end{equation}

Then, the buffer-updating rules based on the controller's and the actuator's operations are
\begin{equation} \label{eq:b_c}
\vec{\mathbf{b}}_c(t) \!=\! \begin{cases}
\mathbf{S}_c^{L(t)}\begin{bmatrix}
 \vec{\mathbf{u}}(t)\\
\mathbf{0}
\end{bmatrix}, &\text{\!\!\!\!\!if } \gamma'(t)N(t)>0, \gamma(t)=1\\
\mathbf{0}, &\text{\!\!\!\!\!if } \gamma'(t)N(t)>0, \gamma(t)=0\\
\mathbf{0}, &\text{\!\!\!\!\!if } \gamma'(t)N(t)=0,\lambda_a(t-1)=0\\
\mathbf{S}_c^{L(t)} \vec{\mathbf{b}}_c(t-1), &\text{\!\!\!\!\!if } \gamma'(t)N(t)=0,\lambda_a(t-1)\!\neq\! 0, \gamma(t)\!=\!1\\
\vec{\mathbf{b}}_c(t-1),&\text{\!\!\!\!\!if } \gamma'(t)N(t)=0,\lambda_a(t-1)\!\neq\! 0, \gamma(t)\!=\!0\\
\end{cases}
\end{equation}
and
\begin{equation} \label{eq:b_a}
\vec{\mathbf{b}}_a(\!t) \!\!=\!\! \begin{cases}
\mathbf{S}_a\vec{\mathbf{b}}_a(t-1), \text{\hspace{2.2cm} if }\gamma(t)=0\\
\!\!\!\begin{bmatrix}
[\mathbf{u}_1(t)^\top,\cdots,\mathbf{u}_{L(t)}(t)^\top]^\top\\
\mathbf{0}
\end{bmatrix}, \text{ if }\gamma'(t)N(t)>0,\gamma(t)=1\\
\mathbf{0}, \begin{aligned}
\text{\hspace{3.6cm} if }\gamma'(t)N(t)=0, 
\vec{\mathbf{b}}_a(t) = \mathbf{0}
\end{aligned}\\
\!\!\!\begin{bmatrix}
\![\mathbf{b}_{a,2}(t-1)^\top,\cdots,\mathbf{b}_{a,\lambda_a(t-1)}(t)^\top]^\top\\
[\mathbf{b}_{c,1}(t)^\top,\cdots,\mathbf{b}_{c,L(t)}(t)^\top]^\top\\
\mathbf{0}
\end{bmatrix}\!\!\!,\! \begin{aligned}
\text{if }\gamma'(t)N(t)\!=\!0, \\
\vec{\mathbf{b}}_a(t) \!\neq\! \mathbf{0}, \gamma(t)\!=\!1.
\end{aligned}\\
\end{cases}
\end{equation}

From \eqref{eq:b_c} and \eqref{eq:b_a}, the updating rules of the effective buffer lengths $\lambda_c(t)$ and $\lambda_a(t)$ are 
\begin{equation}\label{eq:prob_c}
\lambda_c(t) = \begin{cases}
N(t) - L(t) ,& \text{if } \gamma'(t)N(t)>0,\gamma(t) =1\\
0 ,& \text{if } \gamma'(t)N(t)>0,\gamma(t) =0\\
0 ,& \text{if } \gamma'(t)N(t)=0,\lambda_a(t-1) =0\\
\lambda_c(t+1) - L(t) ,& \text{\!\!\!\!if } \gamma'(t)N(t)=0,\lambda_a(t-1) \! \neq \! 0,\! \gamma(t)\!=\!1\\
\lambda_c(t-1),& \text{\!\!\!\!if } \gamma'(t)N(t)=0,\lambda_a(t-1) \!\neq\! 0,\! \gamma(t)\!=\!0
\end{cases}
\end{equation}
and
\begin{equation}\label{eq:prob_a}
\lambda_a(t) \!=\! \begin{cases}
\!\max\{\lambda_a(t\!-\!1)\!-\!1,\!0\} ,& \text{\!if } \gamma(t) =0\\
\!L(t) ,& \text{\!if } \gamma'(t)N(t)>0,\gamma(t) =1\\
\!0 ,& \text{\!if } \gamma'(t)N(t)\!=\!0,\lambda_a(t-1) \!=\!0\\
\!\lambda_a(t-1) + L(t)-1 ,& \begin{aligned}
&\text{\!if } \gamma'(t)N(t)\!=\!0,\lambda_a(t-1) \!\neq\! 0, \\
&\gamma(t)=1.
\end{aligned}
%\lambda_a(t-1)-1,& \text{if } \gamma'(t)N(t)=0,\lambda_a(t-1) \neq 0, \gamma(t)=0
\end{cases}
\end{equation}

\section{Stability of the Anytime Control System}\label{sec:stability}
Based on the anytime control method described in \eqref{eq:V}, \eqref{eq:b_c} and \eqref{eq:b_a}, and following the established stability analysis framework adopting stochastic Lyapunov functions~\cite{anytime2,anytime3,anytime5}, to investigate the stability condition of the system~\eqref{eq:sys}, we only need to focus on the plant events that the actuator runs out of control commands, i.e., $\lambda_a(t) =0, \forall t\in\mathbb{N}_0$.
However, since the process $\{\lambda_a\}_{\mathbb{N}_0}$ has an infinite memory and is not a Markov process, the methods in~\cite{anytime2,anytime3,anytime5} are not directly applicable. Instead,  we shall analyze the control system through the aggregated Markov process $\{Z\}_{\mathbb{N}_0}$ defined as
\begin{equation}\label{eq:Z}
\begin{aligned}
Z(t) &\triangleq \left(\lambda_c(t),\lambda_a(t),B(t+1),B'(t+1),N(t+1)\right)\\
& \hspace{2.6cm } \in \mathcal{X}_c \times \mathcal{X}_a \times \mathcal{B} \times \mathcal{B}' \times\mathcal{N}, t\in \mathbb{N}_0,
\end{aligned}
\end{equation}
where $\mathcal{X}_c \triangleq \{0,1,\cdots,\Lambda_c\}$ and $\mathcal{X}_a \triangleq \{0,1,\cdots,\min\{\Lambda_a,\bar{N}\}\}$.
Assume that $Z(t), \forall t \in\mathbb{N}_0$, belongs to the finite set $\mathcal{S} \triangleq \{s_0,s_1,\cdots,s_{S}\}$ with cardinality $S$.	
Different from~\cite{anytime2}, which only needs to analyze an aggregated process of two processes, we need to investigate the aggregation of five processes, where both $\{\lambda_c\}_{\mathbb{N}_0}$ and $\{\lambda_a\}_{\mathbb{N}_0}$ are correlated with $\{B,B',N\}_{\mathbb{N}_0}$.

Since the control process is divided by the open-loop events with $\lambda_a(t)=0$,
we define $\mathcal{K} = \{k_n\}_{n\in \mathbb{N}_0}$ as the sequence of time steps with $\lambda_a(t)=0$. 
We name the time sequence between $k_n$ and $k_{n+1}$ as the $(n+1)$th cycle of the process, $\forall n \in \mathbb{N}_0$.
Then, the number of time steps between consecutive elements of $\mathcal{K}$ is 
\begin{equation} \label{eq:delta}
\Delta_{n+1} = k_{n+1}-k_n.
\end{equation}

Without loss of generality, let the set $\mathcal{S}_0 \triangleq \{s_0,s_1,\cdots,s_{S_0}\} \in \mathcal{S}$ with cardinality of $S_0$ denote the subset of $\mathcal{S}$ consisting of all the states with $\lambda_a(t)=0$, and hence $Z(k_n) \in \mathcal{S}_0,\forall n\in\mathbb{N}_0$.
In~\cite{anytime2,anytime3,anytime5}, the set $\mathcal{S}_0$ has only one state. In our scenario, $S_0>1$ introduces more challenges in analyzing the process $\{k_n\}_{n\in \mathbb{N}_0}$. 
The state transition process of $\{Z\}_{\mathbb{N}_0}$ is illustrated in Fig.~\ref{fig:set}.
\begin{figure}[t]
	\centering\includegraphics[scale=0.6]{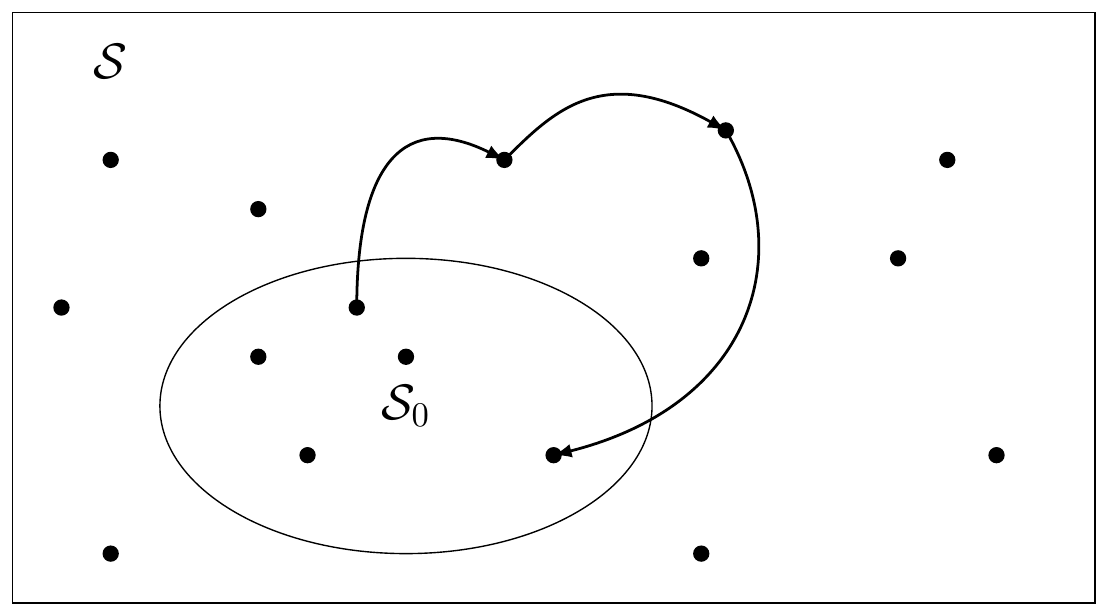}
	\vspace{-0.3cm}	
	\caption{An illustration of the state transition process of $\{Z\}_{\mathbb{N}_0}$.}
	\vspace{-0.5cm}	
	\label{fig:set}
\end{figure}

In what follows, we study the properties of $\{Z\}_{\mathbb{N}_0}$ and $\{Z(k_n)\}_{n\in\mathbb{N}_0}$ and then analyze the stability condition.

\subsection{Properties of $\{Z\}_{\mathbb{N}_0}$ and $\{Z(k_n)\}_{n\in\mathbb{N}_0}$}
For ease of analysis, we need the following assumption about the aggregated process $\{(B,B',N)\}_{\mathbb{N}_0}$.
\begin{assumption}\label{assump:BBN}
	\normalfont
	$\{(B,B',N)\}_{\mathbb{N}_0} \in \mathcal{B} \times \mathcal{B}' \times\mathcal{N}$ is an irreducible and aperiodic (IA) Markov process.
\end{assumption}
Note that wireless fading channel conditions are commonly modeled as IA Markov processes~\cite{Parastoo,KangJIoT,liu2020remote,liu2021remote}, thus it is reasonable to consider an IA Markov process $\{(B,B')\}_{\mathbb{N}_0} \in \mathcal{B} \times \mathcal{B}'$. 
Also the computation availability $\{N\}_{\mathbb{N}_0} \in \mathcal{N}$ is commonly modeled as an IA Markov process~\cite{anytime2}. Since an aggregation of IA Markov processes is still an IA Markov process~\cite{durrett2019probability}, it is reasonable to consider the IA Markov process $\{(B,B',N)\}_{\mathbb{N}_0}$.

For the initial state of $\{Z\}_{\mathbb{N}_0}$, it is practical to assume that both buffers are empty:
\begin{assumption}\label{assump:initial}
	\normalfont	
	Let $\lambda_c(0)=\lambda_a(0)=0$, $B(0)=B_0 \in \mathcal{B}$, $B'(0) = B'_0 \in \mathcal{B}'$ and $N(0)=N_0\in \mathcal{N}$.
	%where there exists $b''$ and $N'$ such that $ \bm{1}(g(b'') >0 )N=0$, $p_{b'',b'}>0$ and $q_{N',N}>0$.
\end{assumption}

\begin{lemma}\label{lem:Z}
	\normalfont
	Consider Assumptions~\ref{assump:BBN} and \ref{assump:initial}. Then $\{Z\}_{\mathbb{N}_0}$ is an IA Markov process.
\end{lemma}
\begin{proof}
(1) Irreducibility. Considering the initial state $Z(0)$, to prove that it is possible to get to any state from any state in $\{Z\}_{\mathbb{N}_0}$, we only need to show that any state $(\tilde{\lambda}_c,\tilde{\lambda}_a,\tilde{B},\tilde{B}',\tilde{N}) \in \mathcal{S}$ can return to $Z(0)$~\cite{durrett2019probability}.
Since $\{(B,B',N)\}_{\mathbb{N}_0}$ is an IA Markov process, we assume that $(\tilde{B},\tilde{B}',\tilde{N})$ can get to $(B_0,B'_0,N_0)$ in $l$ steps with probability $p>0$, where $l>\Lambda_a$. Since $\bar{\gamma}>0$, when the packet dropout event consecutively occurs for $l$ times with probability $\bar{\gamma}^l>0$, the effective buffer lengths will be zero at both controller and actuator. Thus, the state $(\tilde{\lambda}_c,\tilde{\lambda}_a,\tilde{B},\tilde{B}',\tilde{N})$ can return to $Z(0)$ in $l$ steps with a non-zero probability.

(2) Aperiodicity. To prove the aperiodicity of $\{Z\}_{\mathbb{N}_0}$, we only need to show that a single state state $Z(0)$ is aperiodic due to the irreducebility of $\{Z\}_{\mathbb{N}_0}$~\cite{durrett2019probability}.
Due to the aperiodicity of $\{(B,B',N)\}_{\mathbb{N}_0}$, the period of the state $(B_0,B'_0,N_0)$ is $1$ and is written as~\cite{durrett2019probability}
\begin{equation}\label{eq:gcd}
\begin{aligned}
&\gcd\{t>0:   \myprob{(B(t),B'(t),N(t))=(B_0,B'_0,N_0) \vert \right.\\
&\left. \hspace{2.2cm} (B(0),B'(0),N(0))=(B_0,B'_0,N_0)}>0\}=1.
\end{aligned}
\end{equation}
If state $(B_0,B'_0,N_0)$ can reach itself in $t\in \mathbb{N}$ steps, due to the non-zero probability of consecutive packet dropout of $t$ times, the state $(0,0,B_0,B'_0,N_0)$ can reach itself in $t$ steps as well.
%$$\myprob{(\lambda_c(t),\lambda_c(t),B(t+1),B'(t+1),N(t+1))=(0,0,B,B',N) \vert (\lambda_c(0),\lambda_c(0),B(1),B'(1),N(1))=(0,0,B,B',N)}>0.$$
Then, from \eqref{eq:gcd}, the period of the state $(0,0,B_0,B'_0,N_0)$ is $1$. 
\end{proof}

From Lemma~\ref{lem:Z}, $\{Z\}_{\mathbb{N}_0}$ has a unique stationary distribution.
We shall denote the state transition probability of $\{Z\}_{\mathbb{N}_0}$ as per
\begin{equation}\label{eq:v_ij}
\begin{aligned}
v_{i,j}& \triangleq \myprob{Z(t+1)=s_j \vert Z(t)=s_i}, \forall s_i,s_j \in \mathcal{S},  t\in \mathbb{N}_0.
\end{aligned}
\end{equation}
Note that $v_{i,j}$ can be numerically calculated based on \eqref{eq:prob_B}, \eqref{eq:prob_N}, \eqref{eq:prob_c} and \eqref{eq:prob_a}, though it does not have a closed-from expression due to the complexity introduced by the dual-buffer updating process \eqref{eq:prob_c} and \eqref{eq:prob_a}.

Let $\mathbf{V}\in \mathbb{R}^{S\times S}$ denote the state transition probability matrix, i.e., $[\mathbf{V}]_{i,j}=v_{i,j}$, and 
\begin{equation}\label{eq:V_devide}
\mathbf{V} = \begin{bmatrix}
\mathbf{V}_{0,0} &\aug & \mathbf{V}_{0,1} \\ \hline
\mathbf{V}_{1,0} &\aug & \mathbf{V}_{1,1} \\
\end{bmatrix}
\end{equation}
where $\mathbf{V}_{0,0}\in\mathbb{R}^{S_0 \times S_0}$, $\mathbf{V}_{0,1}\in\mathbb{R}^{S_0 \times (S-S_0)}$,
$\mathbf{V}_{1,0}\in\mathbb{R}^{(S-S_0) \times S_0}$ and $\mathbf{V}_{1,1}\in\mathbb{R}^{(S-S_0) \times (S-S_0)}$.

\begin{lemma}\label{lem:transition}
	\normalfont
	$\{Z(k_n)\}_{n\in\mathbb{N}_0}$ is an IA Markov process under Assumptions~\ref{assump:BBN} and \ref{assump:initial}. The state transition probability matrix is 
	
\begin{equation}\notag
\tilde{\mathbf{V}} \triangleq \sum_{l=1}^{\infty} \mathbf{D}(l) \in \mathbb{R}^{S_0 \times S_0},
\end{equation}
where	
	\begin{equation}\notag
	\mathbf{D}(l)=\begin{cases}
	\mathbf{V}_{0,0},& l=1\\
	\mathbf{V}_{0,1} \mathbf{V}_{1,1}^{l-2} \mathbf{V}_{1,0} ,& l>1
	\end{cases}
	\end{equation}
The stationary distribution of $s_i \in \mathcal{S}_0$, $\pi_i$, is the unique solution of
\begin{equation}
\bm{\pi}^\top \tilde{\mathbf{V}} = \bm{\pi}^\top,
\end{equation}
where $\bm{\pi} \triangleq [\pi_1,\pi_2,\cdots,\pi_{S_0}]^\top$.
\end{lemma}
\begin{proof}
	The irreducibility of $\{Z(k_n)\}_{n\in\mathbb{N}_0}$ is obvious, because  $\{Z\}_{\mathbb{N}_0}$ is irreducible. For the aperiodicity, we only need to prove that there exist one state of $\{Z(k_n)\}_{n\in\mathbb{N}_0}$ with period $1$~\cite{durrett2019probability}.
Due to the aperiodicity of $\{(B,B',N)\}_{\mathbb{N}_0}$, the period of the state $(B,B',N)$ is $1$.
If the state $(B,B',N)$ can reach itself in $n\in \mathbb{N}$ steps, due to the non-zero probability of consecutive packet dropout of $n$ times, the state $(0,0,B,B',N) \in \mathcal{S}_0$ can reach itself in $n\in \mathbb{N}$ steps as well without passing through any state within $\mathcal{S}_1 \triangleq \mathcal{S}\backslash \mathcal{S}_0$.
Therefore, the period of the state $(0,0,B,B',N)$ of the process $\{Z(k_n)\}_{n\in\mathbb{N}_0}$ is $1$. In the following, we derive the state transition probability matrix.

We define the conditional probability
\begin{equation}
\begin{aligned}
& d_{i,j}(l)\\
& \!\triangleq \! \myprob{\Delta_{n+1}\!=\!l, Z(k_{n+1})\!=\! s_j\vert Z(k_n)\!=\!s_i},  s_i,s_j\!\in\! \mathcal{S}_0, l \in\! \mathbb{N},
\end{aligned}
\end{equation}
and note that
\begin{equation}\label{eq:d_ij}
d_{i,j}(1) = v_{i,j}, \forall s_{i}, s_{j} \in \mathcal{S}_0.
\end{equation}
Let $\tilde{\Delta}_{n+1}$ denote the number of steps to go from $Z(k_n) = s_i \in \mathcal{S}_0$ to $s_j \in \mathcal{S}_1$ without passing through any states in $\mathcal{S}_0$. We define the following conditional probability
\begin{equation}\label{eq:d'_ij}
\begin{aligned}
&\!\tilde{d}_{i,j}(l) \\
&\!\!\triangleq\!\! \myprob{\!\tilde{\Delta}_{n+1}\!\!=\!l,\!  Z(\!k_n\!+\!\tilde{\Delta}_{n+1}\!)\!\!=\!s_j \!\vert \! Z(\!k_n\!)\!=\!s_i, k_n \!\!+\! \tilde{\Delta}_{n+1}\!<\!k_{n+1}\!}\!,\\
& \hspace{5.3cm} \forall s_i\in \mathcal{S}_0,s_j\in \mathcal{S}_1, l \in \mathbb{N}
\end{aligned}
\end{equation}
and 
\begin{equation}
\tilde{d}_{i,j}(1)= v_{i,j}, \forall s_{i} \in \mathcal{S}_0, s_{j} \in \mathcal{S}_1.
\end{equation}
Then, it can be shown that
\begin{align} 
\label{eq:d_new}
\tilde{d}_{i,k}(l+1) &= \sum_{k' \in \mathcal{S}_1}^{} \tilde{d}_{i,k'}(l) v_{k',k},\\
\label{eq:d}
d_{i,j}(l+1) &= \sum_{k \in \mathcal{S}_1}^{} \tilde{d}_{i,k}(l) v_{k,j}.
\end{align}

Let's introduce $\mathbf{D}(l) \in \mathbb{R}^{S_0 \times S_0}$ and $\tilde{\mathbf{D}}(l) \in \mathbb{R}^{S_0 \times (S-S_0)}$, where $[\mathbf{D}(l)]_{i,j} = d_{i,j}(l)$ and $[\tilde{\mathbf{D}}(l)]_{i,j} = \tilde{d}_{i,j}(l)$.
From \eqref{eq:d_new}, we have
\begin{equation}
\tilde{\mathbf{D}}(l+1)=\tilde{\mathbf{D}}(l) \mathbf{V}_{1,1},
\end{equation}
and hence
\begin{equation} \label{eq:D_new}
\tilde{\mathbf{D}}(l) = \tilde{\mathbf{D}}(1) \mathbf{V}_{1,1}^{l-1}.
\end{equation}
Similarly, from \eqref{eq:d}, we have
\begin{equation}\label{eq:D}
\mathbf{D}(l+1)=\tilde{\mathbf{D}}(l) \mathbf{V}_{1,0}.
\end{equation}
Using \eqref{eq:d_ij}, \eqref{eq:d'_ij}, \eqref{eq:D_new} and \eqref{eq:D}, we have
\begin{equation}
\mathbf{D}(l)=\begin{cases}
\mathbf{D}(1) = \mathbf{V}_{0,0},& l=1\\
\tilde{\mathbf{D}}(1) \mathbf{V}_{1,1}^{l-2} \mathbf{V}_{1,0} 
= \mathbf{V}_{0,1} \mathbf{V}_{1,1}^{l-2} \mathbf{V}_{1,0} ,& l>1.
\end{cases}
\end{equation}
Therefore, the state transition probability matrix satisfies
$
\tilde{\mathbf{V}} \triangleq \sum_{l=1}^{\infty} \mathbf{D}(l).
$
\end{proof}
%Let $\mathbf{V}_0$ denote the state transition probability matrix, i.e., $[\mathbf{V}_0]_{i,j} \triangleq v_{i,j}$. The stationary distribution of $s_i \in \mathcal{S}$, $\pi_i$, can be obtained by solving
%\begin{equation}
%\bm{\pi}^\top \mathbf{V}_0 = \bm{\pi}^\top,
%\end{equation}
%where $\bm{\pi} \triangleq [\pi_1,\pi_2,\cdots,\pi_{S}]^\top$.

\subsection{Analysis of the Stability Condition}
%We name the time sequence between $k_n$ and $k_{n+1}$ as the $(n+1)$th cycle of the process, $\forall n \in \mathbb{N}_0$.
%Then, the amount of time steps between consecutive elements of $\mathcal{K}$ is 
%\begin{equation}
%\Delta_{n+1} = k_{n+1}-k_n.
%\end{equation}
Similar to~\cite{anytime2,anytime3,anytime5}, the stability of the WNCS depends on the statistics of $\{\Delta_n\}_{n\in\mathbb{N}}$ in~\eqref{eq:delta}, which denotes the time duration between consecutive open-loop events. Different to~\cite{anytime2,anytime5}, in the considered case,  
the process $\{\Delta_n\}_{n\in\mathbb{N}}$ is not  i.i.d.  
For WNCS with a single channel, in~\cite{anytime3} an event-triggered  setup was considered, leading to $\{\Delta_n\}_{n\in\mathbb{N}}$ which is not i.i.d. However, our current setup  is different from~\cite{anytime3}. In particular, $\{\Delta_n\}_{n\in\mathbb{N}}$ is formed by the first return time of a set of states in $\mathcal{S}_0$ rather than a single state (as in~\cite{anytime3}) . Therefore, the approach in~\cite{anytime3} cannot be adopted directly.
In the following, we propose a novel cycle-cost-based approach to obtain sufficient stability conditions.

\begin{lemma} \label{lem:Xi}
	\normalfont
The plant is stochastically stable if 
\begin{equation}\label{eq:Xi}
\sum_{n=1}^{\infty} \myexpect{\Xi(n)} <\infty,
\end{equation}
where
\begin{equation} \label{eq:Xi_def}
\Xi(n) = \alpha^n \rho^{\sum_{i=1}^{n} (\Delta_i-1)}, \forall n \in \mathbb{N},
\end{equation}
and $\rho$ and $\alpha$ were defined in {\color{black}Assumptions \ref{assum:rho} and \ref{assum:alpha}}, respectively.
\end{lemma}
\begin{proof}
	From Assumptions~\ref{assum:rho} and \ref{assum:alpha}, and Definition~\ref{def:stability}, the plant is stable if we can prove that 
	\begin{equation}
	\myexpect{\sum_{t=0}^{\infty} V(\mathbf{x}(t))}<\infty.
	\end{equation}
	
 By using Assumptions~\ref{assum:rho} and \ref{assum:alpha}, we have
	\begin{equation}
	\begin{aligned}
	\sum_{t=k_n}^{k_{n+1}-1} V(\mathbf{x}(t))
	&\leq \left(1+ \alpha \sum_{l=0}^{k_{n+1}-k_{n}-2} \rho^l\right) V(\mathbf{x}(k_n))\\
	&< \left(1+ \frac{\alpha}{1-\rho}\right) V(\mathbf{x}(k_n)),
	\end{aligned}
	\end{equation}
and hence
\begin{equation}
\sum_{t=0}^{\infty} V(\mathbf{x}(t))\leq \left(1+ \frac{\alpha}{1-\rho}\right)
\sum_{n=0}^{\infty} V(\mathbf{x}(k_n)).
\end{equation}
Then, it is easy to obtain that 
\begin{equation}
V(\mathbf{x}(k_n)) \leq \alpha^n \rho^{\sum_{i=1}^{n} (\Delta_i-1)} V(\mathbf{x}(0)).
\end{equation}
Since
$
\myexpect{V(\mathbf{x}(k_n))} \leq \alpha^n \myexpect{\rho^{\sum_{i=1}^{n} (\Delta_i-1)} } \myexpect{V(\mathbf{x}(0))}
$ and $\myexpect{V(\mathbf{x}(0))}<\infty$, \eqref{eq:Xi} is proved.
\end{proof}

From Lemma~\ref{lem:Xi}, to find the stability condition of the system, we only need to investigate the process $\{\Xi\}_{\mathbb{N}}$. However, $\{\Xi\}_{\mathbb{N}}$ is not Markovian, whereas the underlying process $\{Z(k_n)\}_{n\in \mathbb{N}_0}$ is. 
In the following, we investigate two stability conditions, with and without exploring the  state transition properties of the underlying process. The results are stated in  Theorems~\ref{theory:loose} and~\ref{theory:tight}, respectively.

Before proceeding, we need the technical lemma below.
\begin{lemma}\label{lem:loose}
	\normalfont
	The following inequality holds
	$$
	\myexpect{\Xi(n+1)} <  \left(\frac{\alpha}{\rho} \max_{i,j\in\{1,\cdots,S_0\}} r_{i,j}\right)	\myexpect{\Xi(n)}, \forall n \in \mathbb{N},
	$$
	where
	\begin{equation}\notag
	\begin{aligned}
	r_{i,j} & \triangleq \myexpect{\rho^{\Delta_{n+1}} \vert Z(k_n)=s_i,Z(k_{n+1})=s_j} \\
	&= \frac{\sum_{l=1}^{\infty} \rho^l [\mathbf{D}(l)]_{i,j}}{[\tilde{\mathbf{V}}]_{i,j}}<\infty, i,j\in\{1,\cdots,S_0\},
	\end{aligned}
	\end{equation}
	and $\mathbf{D}(l)$ and $\tilde{\mathbf{V}}$ are defined in Lemma~\ref{lem:transition}.
\end{lemma}
\begin{proof}
From \eqref{eq:Xi_def}, it can be shown that
	\begin{equation}
	\Xi(n+1) = \Xi(n) \frac{\alpha}{\rho} \rho^{\Delta_{n+1}},
	\end{equation}
	and hence
	\begin{equation}
	\begin{aligned}
	&\myexpect{\Xi(n+1)} = \myexpect{\Xi(n) \frac{\alpha}{\rho} \rho^{\Delta_{n+1}}}\\
	&= \frac{\alpha}{\rho} \sum_{i=1}^{S_0}\sum_{j=1}^{S_0} \myexpect{\Xi(n)  \rho^{\Delta_{n+1}} \vert Z(k_n)=s_i, Z(k_{n+1})=s_j} \\
	& \hspace{4cm} \times \myprob{Z(k_n)=s_i, Z(k_{n+1})=s_j}\\
	&\leq  \frac{\alpha}{\rho} \max_{i,j\in\{1,\cdots,S_0\}} r_{i,j} \sum_{i=1}^{S_0}\sum_{j=1}^{S_0} \myexpect{\Xi(n)  \vert Z(k_n)=s_i, Z(k_{n+1})=s_j}\\ &\hspace{4cm} \times \myprob{Z(k_n)=s_i, Z(k_{n+1})=s_j}\\
	&= \frac{\alpha}{\rho}  \max_{i,j\in\{1,\cdots,S_0\}} r_{i,j}  \myexpect{\Xi(n) },
	\end{aligned}
	\end{equation}
where
\begin{equation}\label{eq:r_ij}
\begin{aligned}
r_{i,j}& \triangleq \myexpect{\rho^{\Delta_{n+1}} \vert Z(k_n)=s_i,Z(k_{n+1})=s_j}\\
&=\sum_{\Delta_{n+1}=1}^{\infty}\rho^{l}\frac{\myprob{\Delta_{n+1}=l, Z(k_{n+1})=s_j \vert  Z(k_n)=s_i}}{\myprob{Z(k_{n+1})=s_j \vert Z(k_n)=s_i}}\\
&= \frac{\sum_{l=1}^{\infty} \rho^l [\mathbf{D}(l)]_{i,j}}{[\tilde{\mathbf{V}}]_{i,j}},
\end{aligned}
\end{equation}
and $r_{i,j}$ is bounded due to the fact that $\sum_{l=1}^{\infty} [\mathbf{D}(l)]_{i,j}<1$ and $\rho <1$.	
\end{proof}
Using Lemmas~\ref{lem:Xi} and~\ref{lem:loose}, it is straightforward to have the following result.
\begin{theorem}\label{theory:loose}
\normalfont	
	Suppose that Assumptions~\ref{assum:rho}-\ref{assump:initial} hold. The system~\eqref{eq:sys} is stochastically stable if 
	$$\Omega' \triangleq \frac{\alpha}{\rho}\max_{i,j\in\{1,\cdots,S_0\}} r_{i,j} < 1.$$
\end{theorem}
\begin{remark}
For the special case with $S_0 = 1$, i.e., there exists only a single state in $\mathcal{S}$ leading to no control being applied, Lemma~\ref{lem:transition} shows that the processes $\{Z(k_n)\}_{n\in\mathbb{N}_0}$ is i.i.d.  and so is the process $\{\Delta_n\}_{n\in\mathbb{N}}$. Such a special case is identical to that of \cite{anytime2}, and we see that
the stability condition in Theorem~\ref{theory:loose} reduces to
\begin{equation}
\Omega' = \frac{\alpha}{\rho} \sum_{l=1}^{\infty} \rho^l \myprob{\Delta_n=l} < 1,
\end{equation}
which is identical to that of~\cite{anytime2}.
\end{remark}
\begin{remark}
The sufficient condition in Theorem \ref{theory:loose} is obtained by considering the worst case scenario of the state transitions, i.e., considering the pair of states $s_i$ and $s_j$ such that it has the largest conditional expectation $\myexpect{\rho^{\Delta_{n+1}} \vert Z(k_n)=s_i,Z(k_{n+1})=s_j}$.
However, such a method does not take into account the state transition probabilities in  $\mathbf{\tilde{V}}$ defined in Lemma~\ref{lem:transition}. Thus, the sufficient condition in Theorem~\ref{theory:loose} is conservative.
\end{remark}

The following lemma is needed to obtain a less conservative stability condition in Theorem~\ref{theory:tight}.
\begin{lemma} \label{lem:Omega}
\normalfont
Assuming that $\{Z\}_{\mathbb{N}_0}$ evolves in the steady state, for any arbitrarily small $\epsilon>0$, there exist $\mu>0$ and $K>0$ such that
\begin{equation}\label{eq:lem:key} 
\myexpect{\Xi(n)} < \mu \Big(\frac{\alpha}{\rho} \left(\lambda_{\max}(\mathbf{U})+\epsilon\right)\Big)^{n-1}, \forall n> K,
\end{equation}
where
$
\mathbf{U} \triangleq
\begin{bmatrix}
\mathbf{U}_1^\top &\aug &
\mathbf{U}_2^\top &\aug & 
\cdots &\aug& 
\mathbf{U}_{S_0}^\top 
\end{bmatrix}^\top\in \mathbb{R}^{S^2_0 \times S^2_0},
$
\begin{equation}  \notag
\mathbf{U}_i \triangleq
\begin{bmatrix}
r_{1,i} \mathbf{f}^\top_1& \mathbf{0} &\cdots &\mathbf{0}\\
\mathbf{0} & r_{2,i} \mathbf{f}^\top_2 &\cdots &\mathbf{0}\\
\vdots & \cdots &\ddots & \vdots \\
\mathbf{0} & \cdots &\mathbf{0} & r_{S_0,i} \mathbf{f}^\top_{S_0}
\end{bmatrix}  \in \mathbb{R}^{S_0 \times S^2_0},
\end{equation}
and $\begin{bmatrix}
\mathbf{f}_1 &\aug & \mathbf{f}_2&\aug & \cdots &\aug&\mathbf{f}_{S_0}
\end{bmatrix} \triangleq \mathbf{Z} \tilde{\mathbf{V}} \mathbf{Z}^{-1}$ and  $\mathbf{Z} = \text{diag}\{\bm{\pi}\}$.
\end{lemma}
\begin{proof}
See Appendix A.
\end{proof}

From Lemmas~\ref{lem:Xi} and~\ref{lem:Omega}, it is straightforward to derive the following result.
\begin{theorem} \label{theory:tight}
	\normalfont
Suppose that Assumptions~\ref{assum:rho}-\ref{assump:initial} hold. System~\eqref{eq:sys} is stochastically stable if 
$$\Omega \triangleq \frac{\alpha}{\rho}\lambda_{\max}(\mathbf{U}) < 1.$$
\end{theorem}
\begin{remark}
This stability condition has considered the effect of all the conditional expectations $\myexpect{\rho^{\Delta_{n+1}} \vert Z(k_n)=s_i,Z(k_{n+1})=s_j}$ and the state transition probability matrix of $\{Z(k_n)\}_{n\in\mathbb{N}}$, i.e., $\tilde{\mathbf{V}}$.
Using Perron-Frobenius theorem~\cite{Perron}, since $\mathbf{U}$ is non-negative, we have $\lambda_{\max}(\mathbf{U})\leq \max_{i\in\{1,\cdots,S_0\}}\sum_{j=1}^{S_0}[\mathbf{U}]_{i,j} \leq \max_{i,j\in\{1,\cdots,S_0\}}r_{i,j}$. Thus, the sufficient condition Theorem~\ref{theory:tight} is less restrictive  than that in Theorem~\ref{theory:loose}.
\end{remark}

\subsection{Numerical Example for Stability Condition Calculation}
{\color{black}
We set the buffer lengths $\Lambda_c = \Lambda_a =2$, the maximum number of calculated control commands per time step $\bar{N} = 2$, the maximum number of commands that can be transmitted via the C-A channel $\bar{B} = 2$, the number of channel states of the S-C channel $\bar{B}'=2$, and the packet drop probabilities in the two states  as $0.2$ and $0.01$.
We assume that $\rho = 0.8$.
The Markov state transition probability matrices of the C-A and S-C channel state $(B(t),B'(t)) \in \{0,1,2\}\times \{1,2\}$, and the processor's computational resource $N(t)\in \{0,1,2\}$ are 
\begin{equation}
\mathbf{M} = \begin{bmatrix}
0.24&	0.16&	0.06&	0.04&	0.30&	0.20\\
0.04&	0.36&	0.01&	0.09&	0.05&	0.45\\
0.12&	0.08&	0.06&	0.04&	0.42&	0.28\\
0.02&	0.18&	0.01&	0.09&	0.07&	0.63\\
0&         0&    0.30 &   0.20  &  0.30 &   0.20\\
0 &        0&    0.05 &   0.45  &  0.05  &  0.45
\end{bmatrix},
\end{equation}
and
\begin{equation}
\mathbf{N} = \begin{bmatrix}
0.1 & 0.2 & 0.7\\
0 & 0.6 & 0.4\\
0.1 & 0.3 & 0.6
\end{bmatrix}.
\end{equation}

Thus, the process $Z(t)$ defined in~\eqref{eq:Z} has $3 \times 3 \times 3\times 2\times3 =162$ states.
Using the state transition rules of $\lambda_c(t)$ and $\lambda_a(t)$ in \eqref{eq:prob_c} and \eqref{eq:prob_a}, and the channel and computational resource state transition probabilities~\eqref{eq:prob_B} and \eqref{eq:prob_N}, the state transition matrix of $Z(t)$ can be calculated.
Due to the space limitation, the matrix is presented in~\cite{mymatrix} with $72$ transient states (highlighted in yellow), and thus $Z(t)$ has $S= 162-72=90$ recurrent (effective) states including $S_0=54$ states with $\lambda_a(t) =0$ (highlighted in red).
By removing the transient states, the $90 \times 90$ state transition matrix $\mathbf{V}$ of the effective states can be easily obtained.
Due to the space limitation, it is not possible to show the calculation of the stability conditions based on the $90 \times 90$ matrix $\mathbf{V}$. So we only present the comparison result of the stability conditions in Theorems~\ref{theory:loose} and~\ref{theory:tight}: 
\begin{equation}
\max_{i,j\in\{1,\cdots,S_0\}} r_{i,j} = 0.8
\gg
\lambda_{\max}(\mathbf{U}) = 0.0016,
\end{equation}
showing that the sufficient stability condition in Theorem~\ref{theory:tight} is less restrictive than Theorem~\ref{theory:loose}.

For ease of illustration, we will show the calculation of the stability conditions based on a randomly generated (small) $\mathbf{V}$ with $S=4$ and assume that $S_0 =2$, where
\begin{equation}
\mathbf{V} =
\begin{bmatrix}
0.10&	0.10& \aug&	0.10&	0.70\\
0.30&	0.20& \aug&	0.10&	0.40\\
\hline
0.60&	0.20& \aug&	0.10&	0.10\\
0.90&	0.05& \aug&	0.02&	0.03
\end{bmatrix}
\in \mathbb{R}^{4\times4},
\end{equation}
and thus $\mathbf{V}_{0,0}$, $\mathbf{V}_{0,1}$, $\mathbf{V}_{1,0}$ and $\mathbf{V}_{1,1}$ are obtained directly from \eqref{eq:V_devide}.

Using Lemma~\ref{lem:transition}, the state transition matrix of $\{Z(k_n)\}_{n\in\mathbb{N}_0}$ and its stationary distribution are obtained as
\begin{equation}\label{eq:V_tilder}
\tilde{\mathbf{V}} = \begin{bmatrix}
0.8378&	0.1622\\
0.7546&	0.2454
\end{bmatrix},
\end{equation}
and 
\begin{equation}\label{eq:pi}
\bm{\pi} = [0.8231\ 0.1769]^\top.
\end{equation}

Taking \eqref{eq:V_tilder} and \eqref{eq:pi} into \eqref{eq:r_ij}, we have
$r_{i,j} = [\mathbf{R}]_{i,j}$, where
\begin{equation}\label{eq:R}
\mathbf{R} = \begin{bmatrix}
0.6511&	0.7323\\
0.6971&	0.7673
\end{bmatrix}.
\end{equation}

Taking \eqref{eq:V_tilder}, \eqref{eq:pi} and \eqref{eq:R} into the definition of $\mathbf{U}$ in Lemma~\ref{lem:Omega}, we have
\begin{equation}\label{eq:U}
\mathbf{U} =\begin{bmatrix}
0.3695&	0.0716&	0&	0\\
0&	0&	0.0165&	0.0054\\
0.4156&	0.0805&	0&	0\\
0& 0&	0.0181&	0.0059
\end{bmatrix}.
\end{equation}

By applying \eqref{eq:R} and \eqref{eq:U} onto Theorems~\ref{theory:loose} and~\ref{theory:tight}, respectively, we have
\begin{equation}
\max_{i,j\in\{1,\cdots,S_0\}} r_{i,j} = 0.7673
>
\lambda_{\max}(\mathbf{U}) = 0.3731.
\end{equation}}

\section{Robustness to Process Noise} \label{sec:robust}
In this section, we investigate the stability condition of the plant system below with process noise:
\begin{equation}\label{eq:sys_noise}
\mathbf{x}(t+1) = f(\mathbf{x}(t),\mathbf{u}(t),\mathbf{w}(t)),
\end{equation}
where $\mathbf{w}(t)\in \mathbb{R}^{l_s}$ is a white noise process, which is independent with the other random processes of the system.

For ease of analysis, we consider  uniform bounds and continuity as follows.
\begin{assumption}[\!\!\cite{anytime2}]\label{assum:robust}
	\normalfont
There exists $\beta_x$, $\beta_u$, $\beta_w$, $\beta_V$, $\beta_{\kappa}$, $\rho$, $\alpha$ and $\eta \in \mathbb{R}_{>0}$, such that, $\forall \mathbf{x}, \mathbf{z},\mathbf{w} \in \mathbb{R}^{l_s}$ and $\forall \mathbf{u},\mathbf{v}\in \mathbb{R}^{l_u}$ the following are satisfied
\begin{equation}
\begin{aligned}
& |f(\mathbf{x,u,w}) - f(\mathbf{z,v,0})| \leq \beta_x |\mathbf{x-z}| + \beta_u |\mathbf{u-v}| + \beta_w |\mathbf{w}|\\
& |V(\mathbf{x}) - V(\mathbf{z})| \leq \beta_V |\mathbf{x-z}|\\
& |\kappa(\mathbf{x}) - \kappa(\mathbf{z})| \leq \beta_{\kappa} |\mathbf{x-z}|\\
& V(f(\mathbf{x},\kappa(\mathbf{x}),\mathbf{w})) \leq \rho V(\mathbf{x}) + \eta \vert \mathbf{w} \vert\\
& V(f(\mathbf{x},\mathbf{0},\mathbf{w})) \leq \alpha V(\mathbf{x}) + \eta \vert \mathbf{w} \vert.
\end{aligned}
\end{equation}
\end{assumption}

When considering unbounded process noise, the stability condition in Definition~\ref{def:stability} cannot be satisfied. Thus, we consider the following  stability condition in terms of the average cost~\cite{anytime2}.

\begin{definition}\label{def:stability_rubust}
	\normalfont	
	The dynamical system with process noise \eqref{eq:sys_noise} is stochastically stable, if for some $\psi \in \mathcal{K}_\infty$, the average expected value $\limsup\limits_{T\rightarrow \infty}\frac{1}{T}\sum_{k=0}^{T} \myexpect{\psi(|\mathbf{x}(k)|)}<\infty$.
\end{definition}

\begin{theorem}\label{theory:robust}
\normalfont	
	Suppose that Assumptions~\ref{assum:rho}-\ref{assum:robust} holds. The system~\eqref{eq:sys_noise} is stochastically stable if 
	$\Omega < 1$.
\end{theorem}
\begin{proof}
See Appendix B.
\end{proof}
\begin{remark}
Theorem~\ref{theory:robust} shows that the stability condition for the process-noise-free case holds for the process-noise-present one as well under Assumption~\ref{assum:robust}, which is in line with~\cite{anytime2}.
Note that in the case of Markov jump linear systems, stability conditions for noise-free cases are equivalent to conditions for noisy-cases as well, see Theorem 3.33 of Chapter 3 of \cite{costa2006discrete}.
\end{remark}

\section{Simulation Results}
{\color{black}
We present simulation results of the proposed dual-buffer anytime control algorithm and the baseline single-buffer algorithm~\cite{anytime2}, where the S-C and C-A channels are modeled as multi-state Markov chains. 
The system parameters are the same as in the numerical example in~Section~\ref{sec:stability}.
We consider an open-loop unstable constrained plant model of~\eqref{eq:sys_noise} in the presence of noise~\cite{anytime3}:
\begin{equation}
\begin{bmatrix}
x_1(t)\\
x_2(t)
\end{bmatrix}
=
\begin{bmatrix}
x_2(t) + u_1(t)\\
-\text{sat}(x_1(t)+x_2(t)) + u_2(t)
\end{bmatrix}
+
\begin{bmatrix}
w_1(t)\\
w_2(t)
\end{bmatrix},
\end{equation}
where
\begin{equation}
\text{sat}(\mu) = 
\begin{cases}
-10, \text{ if } \mu<-10\\
\mu, \text{ if } \mu \in [-10,10]\\
10, \text{ if } \mu>10
\end{cases}
\end{equation}
and the process noise $w_1(t)$ and $w_2(t)$ are independent zero-mean i.i.d. Gaussian and $\myexpect{w_1(t)^2}=\myexpect{w_2(t)^2} = 0.1$.
We take the control policy as $\kappa(\mathbf{x}) = [-x_2,\ 0.505 \text{sat}(x_1+x_2)]^\top$.

In Fig.~\ref{fig:simu}, we show the simulation of $|\mathbf{x}(t)|$ with $800$ time steps for both the dual-buffer anytime method and the baseline method. It is clear that the dual-buffer method significantly outperforms the baseline especially when the initial process state $\mathbf{x}(0)$ is far from the origin.
Therefore, adding the command buffer at the actuator can effectively overcome the effect of  packet dropouts, leading to significant control performance improvement.

\begin{figure}[t]
	\centering\includegraphics[scale=0.6]{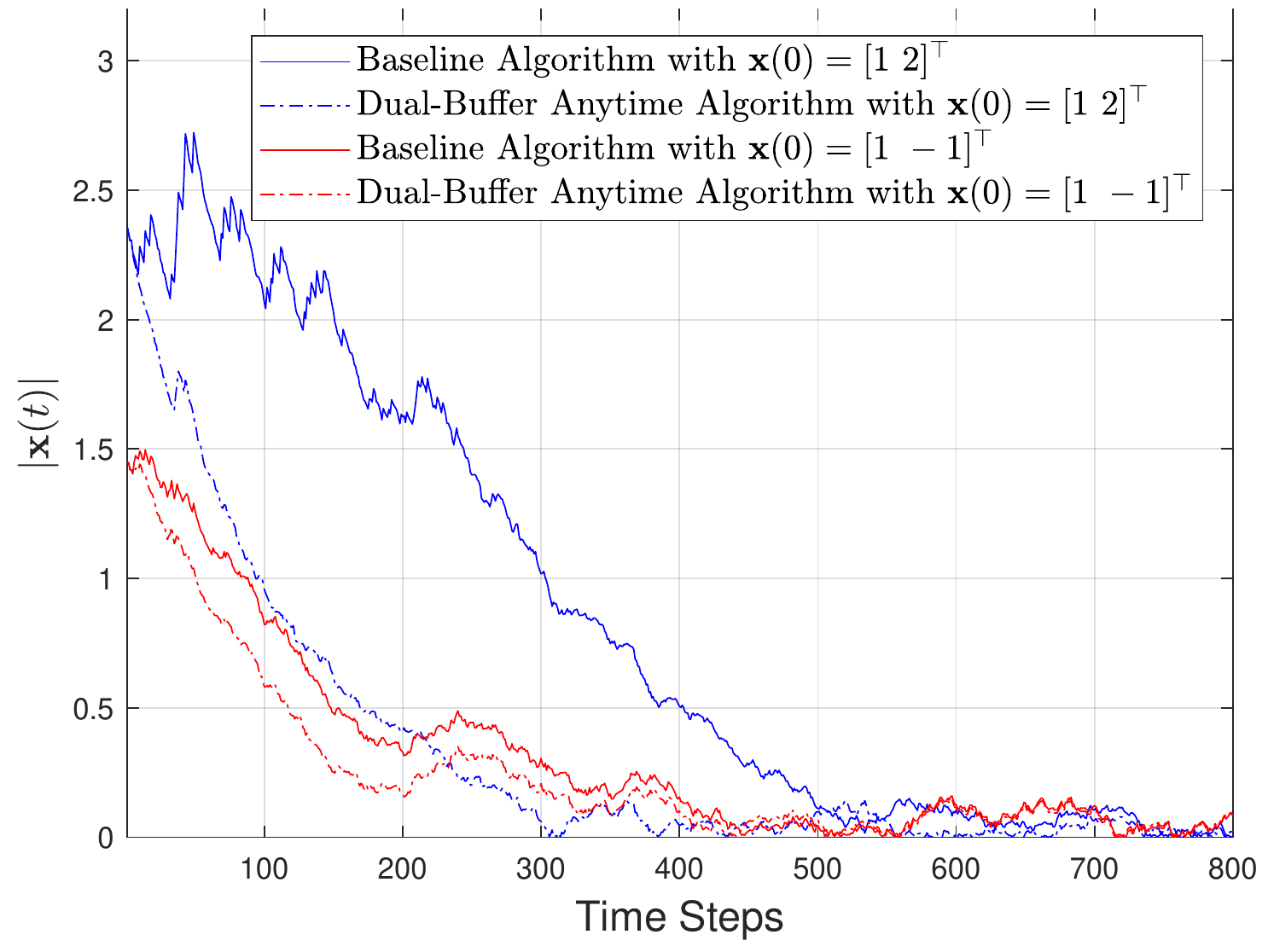}
	\vspace{-0.3cm}	
	\caption{Simulation of $|\mathbf{x}(t)|$ with different initiations.}
	\vspace{-0.5cm}	
	\label{fig:simu}
\end{figure}
}

\section{Conclusions}
We have studied an anytime control algorithm for  dual-channel-dual-buffer WNCS with random computational resources over correlated channels. 
%This class of systems has not been studied in previous literature. 
We have proposed a novel approach to derive sufficient conditions for stochastic  stability in the case of nonlinear plant models with disturbances. 
The stability conditions are stated in terms of plant dynamics, network dynamics,  buffer properties and computational resource dynamics.
Our numerical results have shown that the proposed dual-buffer anytime control system can provide significant control performance improvement when compared to the conventional single-buffer system as it effectively overcomes the effect of packet dropouts.
In future work, we will consider optimal control sequence design.
In addition, we may include an extension to large-scale WNCSs with multiple plants and controllers and that need to share network and communication resources.


%
\appendix
\section*{A. Proof of Lemma~\ref{lem:Omega}}
Direct calculations give that
\begin{equation}\label{eq:main}
\begin{aligned}
&\myexpect{\Xi(n+1) \vert Z(k_n)=s_i ,Z(k_{n+1})=s_j}\\
&= \myexpect{\Xi(n) \frac{\alpha}{\rho} \rho^{\Delta_{n+1}} \vert Z(k_n)=s_i,Z(k_{n+1})=s_j}\\
&= \alpha' \myexpect{\Xi(n) \vert Z(k_n)=s_i}
\myexpect{\rho^{\Delta_{n+1}} \vert Z(k_n)=s_i,Z(k_{n+1})=s_j}\\
&= \alpha'  \myexpect{\Xi(n) \vert Z(k_n)=s_i} r_{i,j}\\
&= \alpha'  r_{i,j} \sum_{i'=1}^{S_0} \myexpect{\Xi(n) \vert Z(k_{n-1})=s_{i'}, Z(k_n)=s_i}\\ & \hspace{3.5 cm}\times \myprob{Z(k_{n-1})=s_{i'} \vert Z(k_n)=s_i} \\
&= \alpha'  r_{i,j} \sum_{i'=1}^{S_0} \myexpect{\Xi(n) \vert Z(k_{n-1})=s_{i'}, Z(k_n)=s_i} {f_{i',i}},
\end{aligned}
\end{equation}
where $\alpha' \triangleq \frac{\alpha}{\rho}$ and
\begin{equation} \label{eq:f}
f_{i',i} \!\triangleq\!\! \myprob{Z(k_{n-1})\!=\!s_{i'} \vert Z(k_n)\!=\!s_i} \!=\! \frac{\pi_{i'}}{\pi_{i}} \tilde{v}_{i',i},\! \forall  i',i \!=\! 1,\cdots,\! S_0.
\end{equation}
From \eqref{eq:f}, we define   $\mathbf{F}\in \mathbb{R}^{S_0 \times S_0}$ such that $[\mathbf{F}]_{i,j} = f_{i,j}$, i.e., 
\begin{equation}
\mathbf{F} \triangleq \begin{bmatrix}
\mathbf{f}_1 &\aug & \mathbf{f}_2&\aug & \cdots &\aug&\mathbf{f}_{S_0}
\end{bmatrix} = \mathbf{Z} \tilde{\mathbf{V}} \mathbf{Z}^{-1} ,
\end{equation}
where $\mathbf{Z} = \text{diag}\{\bm{\pi}\}$ and $\mathbf{f}_{i} \triangleq [f_{1,i},f_{2,i},\cdots,f_{S_0,i}]^\top \in \mathbb{R}^{S_0}$.

%$\mathbf{\pi} \triangleq [\pi_1,\pi_2,\cdots,\pi_{S_0}]$ is the stationary distribution, i.e., $\mathbf{\pi} \mathbf{G} = \mathbf{\pi} $.

Let's introduce   $\mathbf{\Xi}(n) \in \mathbb{R}^{S_0 \times S_0},\forall n\in \mathbb{N}$, where 
\begin{equation}
\begin{aligned}
&\left[\bm{\Xi}(n)\right]_{i,j} \triangleq \myexpect{\Xi(n) \vert Z(k_{n-1})=s_i ,Z(k_{n})=s_j}
\end{aligned}
\end{equation}
Let $\bm{\xi}_i(n) \triangleq \left[[\mathbf{\Xi}(n)]_{1,i},[\mathbf{\Xi}(n)]_{2,i},\cdots,[\mathbf{\Xi}(n)]_{S_0,i}\right]^\top$.
From \eqref{eq:main}, we have
\begin{equation}
\begin{aligned}
&\mathbf{\Xi}(n+1)= \\
&  \!\begin{bmatrix}
\!\alpha' {r_{1,1}} \mathbf{f}^\top_1 \bm{\xi}_1(n) & \alpha' {r_{1,2}} \mathbf{f}^\top_1 \bm{\xi}_1(n) & \cdots &\alpha' {r_{1,S_0}} \mathbf{f}^\top_1 \bm{\xi}_1(n)\\
\!\alpha' {r_{2,1}} \mathbf{f}^\top_2 \bm{\xi}_2(n) & \alpha' {r_{2,2}} \mathbf{f}^\top_2 \bm{\xi}_2(n) & \cdots &\alpha' {r_{2,S_0}} \mathbf{f}^\top_2 \bm{\xi}_2(n)\\
\!\vdots & \vdots & \ddots & \vdots \\
\!\alpha' {r_{S_0,1}} \mathbf{f}^\top_{S_0} \bm{\xi}_{S_0}(n) & \alpha' {r_{{S_0},2}} \mathbf{f}^\top_{S_0} \bm{\xi}_{S_0}(n) & \cdots &\alpha' {r_{{S_0},{S_0}}} \mathbf{f}^\top_{S_0} \bm{\xi}_{S_0}(n)
\end{bmatrix}\!\!.
\end{aligned}
\end{equation}
After some algebraic transformation, it can be obtained that
\begin{equation}\label{eq:v}
\mathbf{v}(n+1) = \mathbf{AU} \mathbf{v}(n),
\end{equation}
where 
\begin{equation}
\mathbf{v}(n) \triangleq [\bm{\xi}_1(n)^\top,\bm{\xi}_2(n)^\top,\cdots,\bm{\xi}_{S_0}(n)^\top]^\top \in \mathbb{R}^{S^2_0},\forall n\in \mathbb{N},
\end{equation}
\begin{equation}
\mathbf{A} \triangleq \text{diag}\{\alpha',\cdots,\alpha'\}\in \mathbb{R}^{S^2_0},
\end{equation}
\begin{equation}
\mathbf{U} \triangleq
\begin{bmatrix}
\mathbf{U}_1 \\ 
\mathbf{U}_2 \\ 
\vdots \\ 
\mathbf{U}_{S_0} \\
\end{bmatrix}\in \mathbb{R}^{S^2_0 \times S^2_0},
\end{equation}
and $\mathbf{U}_i, i =1,\cdots,S_0$, is defined under \eqref{eq:lem:key}.
%\begin{equation}
%\mathbf{U}_i \triangleq
%\begin{bmatrix}
%r_{1,i} \mathbf{f}^\top_1& \mathbf{0} &\cdots &\mathbf{0}\\
%\mathbf{0} & r_{2,i} \mathbf{f}^\top_2 &\cdots &\mathbf{0}\\
%\vdots & \cdots &\ddots & \vdots \\
%\mathbf{0} & \cdots &\mathbf{0} & r_{S_0,i} \mathbf{f}^\top_{S_0}
%\end{bmatrix}  \in \mathbb{R}^{S_0 \times S^2_0}.
%\end{equation}

From \eqref{eq:v}, it follows that
\begin{equation} \label{eq:v_iter}
\mathbf{v}(n) = \mathbf{A}^{n-1} \mathbf{U}^{n-1} \mathbf{v}(1).
\end{equation}
% it is clear that each element of $\mathbf{v}(n)$ decreases exponentially fast when \begin{equation}
%\Omega= \lambda_{\max}(\mathbf{AU}) = \alpha'\lambda_{\max}(\mathbf{U})<1.
%\end{equation}
From~\cite[Lemma~2]{liu2020remote}, for any arbitrarily small $\epsilon>0$, there exists $\mu'>0$ and $K>0$, such that
\begin{equation} \label{eq:Un}
\max_{i,j\in\{1,\cdots,S_0\}}  [\mathbf{U}^{n}]_{i,j} < \mu' \left(\lambda_{\max} (\mathbf{U})+\epsilon\right)^n, \forall 
n> K.
\end{equation}
Taking \eqref{eq:Un} into \eqref{eq:v_iter}, we have
\begin{equation}
\myexpect{\Xi(n)} \! \leq \!\!\!\!\! \max_{i\in\{1,\cdots,S^2_0\}} [\mathbf{v}(n)]_{i}\! < \!\mu \Big(\!\alpha'\left(\!\lambda_{\max}(\mathbf{U})+\epsilon\!\right)\!\!\Big)^{n-1}, \forall \epsilon\!>\!0,
n\!>\! K,
\end{equation}
where $\mu= \mu' S^2_0 \max_{i\in\{1,\cdots,S^2_0\}}[\mathbf{v}(1)]_{i}<\infty$.
This competes the proof.

\section*{B. Proof of Theorem~\ref{theory:robust}}
The proof follows the same steps of~\cite[Theorem 7.1]{anytime2} in analyzing the noise effect on the performance, since the analysis only depends on the process $\{Z(k_n)\}_{n\in {\mathbb{N}_0}}$. From \cite{anytime2}, it can be obtained that
\begin{equation} \label{eq:robust_v_sum}
\begin{aligned}
&\myexpect{\sum_{t=k_{n-1}}^{k_n-1}V(\mathbf{x}(t))} \leq \frac{1+\alpha-\rho}{1-\rho} \myexpect{\Xi(n)} \myexpect{V(\mathbf{x}(0))} \\
&+ \frac{1+\alpha-\rho}{1-\rho} \left(\myexpect{\Xi(n-1)} + \myexpect{\Xi(n-2)}+ \cdots + 1\right) \bar{w} +  \bar{w}',
\end{aligned}
\end{equation}
{\color{black}where $\Xi(n)$ is defined in \eqref{eq:Xi_def}, $\bar{w} = \Psi_1 \myexpect{|\mathbf{w}(t)|}$, $\bar{w}' = \Psi_2 \myexpect{|\mathbf{w}(t)|}$, and $\Psi_1$ and $\Psi_2$ are bounded constant determined by $\beta_x$, $\beta_u$, $\beta_w$, $\beta_V$, $\beta_{\kappa}$, $\rho$, and $\eta$ (see \cite[Appendix D]{anytime2} for the calculation of $\Psi_1$ and $\Psi_2$).}

Based on \eqref{eq:robust_v_sum}, Theorem~\ref{theory:robust} holds if we can prove 
\begin{equation}\label{eq:robust_expect}
\limsup\limits_{n\rightarrow \infty}\sum_{i=1}^{n} \myexpect{\Xi(i)}<\infty.
\end{equation}
Taking \eqref{eq:lem:key} into \eqref{eq:robust_expect}, for any arbitrarily small $\epsilon>0$, there exists $\mu',K'>0$ such that, $\forall n>K'$, the following are satisfied 
\begin{equation}
\begin{aligned}
\sum_{i=1}^{n} \myexpect{\Xi(i)} &\leq \mu' \left(\Big(\frac{\alpha}{\rho} \left(\lambda_{\max}(\mathbf{U})+\epsilon\right)\Big)^{n-2} \right.\\
&\left. \hspace{0.5cm}+\Big(\frac{\alpha}{\rho} \left(\lambda_{\max}(\mathbf{U})+\epsilon\right)\Big)^{n-3}+\cdots+1\right) +c,
\end{aligned}
\end{equation}
where $c$ is a constant.
Then, it is clear that \eqref{eq:robust_expect} holds if $\frac{\alpha}{\rho} \left(\lambda_{\max}(\mathbf{U})\right)<1$. 
%This completes the proof of Theorem~\ref{theory:robust}.

  \balance
    
	\ifCLASSOPTIONcaptionsoff
	\newpage
	\fi


\begin{thebibliography}{10}
	\providecommand{\url}[1]{#1}
	\csname url@samestyle\endcsname
	\providecommand{\newblock}{\relax}
	\providecommand{\bibinfo}[2]{#2}
	\providecommand{\BIBentrySTDinterwordspacing}{\spaceskip=0pt\relax}
	\providecommand{\BIBentryALTinterwordstretchfactor}{4}
	\providecommand{\BIBentryALTinterwordspacing}{\spaceskip=\fontdimen2\font plus
		\BIBentryALTinterwordstretchfactor\fontdimen3\font minus
		\fontdimen4\font\relax}
	\providecommand{\BIBforeignlanguage}[2]{{%
			\expandafter\ifx\csname l@#1\endcsname\relax
			\typeout{** WARNING: IEEEtran.bst: No hyphenation pattern has been}%
			\typeout{** loaded for the language `#1'. Using the pattern for}%
			\typeout{** the default language instead.}%
			\else
			\language=\csname l@#1\endcsname
			\fi
			#2}}
	\providecommand{\BIBdecl}{\relax}
	\BIBdecl
	
	\bibitem{McGovern}
	L.~K. {McGovern} and E.~{Feron}, ``Requirements and hard computational bounds
	for real-time optimization in safety-critical control systems,'' in
	\emph{Proc. IEEE CDC}, 1998, pp. 3366--3371.
	
	\bibitem{McGovern1}
	------, ``Closed-loop stability of systems driven by real-time, dynamic
	optimization algorithms,'' in \emph{Proc. IEEE CDC}, 1999, pp. 3690--3696.
	
	\bibitem{Henriksson}
	D.~{Henriksson}, A.~{Cervin}, J.~{Akesson}, and K.~. {Arzen}, ``On dynamic
	real-time scheduling of model predictive controllers,'' in \emph{Proc. IEEE
		CDC}, 2002, pp. 1325--1330.
	
	\bibitem{ondemand1}
	A.~{Cervin}, M.~{Velasco}, P.~{Marti}, and A.~{Camacho}, ``Optimal online
	sampling period assignment: Theory and experiments,'' \emph{IEEE Trans.
		Control Syst. Technol.}, vol.~19, no.~4, pp. 902--910, 2011.
	
	\bibitem{ondemand2}
	P.~{Tabuada}, ``Event-triggered real-time scheduling of stabilizing control
	tasks,'' \emph{IEEE Trans. Autom. Control}, vol.~52, no.~9, pp. 1680--1685,
	2007.
	
	\bibitem{ondemand3}
	X.~{Wang} and M.~D. {Lemmon}, ``Self-triggered feedback control systems with
	finite-gain $l_2$ stability,'' \emph{IEEE Trans. Autom. Control}, vol.~54,
	no.~3, pp. 452--467, 2009.
	
	\bibitem{bhattacharya2004anytime}
	R.~Bhattacharya and G.~J. Balas, ``Anytime control algorithm: Model reduction
	approach,'' \emph{J. Guidance, Control, Dynamics}, vol.~27, no.~5, pp.
	767--776, Sep.-Oct. 2004.
	
	\bibitem{GuptaAnytime}
	V.~{Gupta} and F.~{Luo}, ``On a control algorithm for time-varying processor
	availability,'' \emph{IEEE Trans. Autom. Control}, vol.~58, no.~3, pp.
	743--748, 2013.
	
	\bibitem{anytime5}
	D.~E. {Quevedo} and V.~{Gupta}, ``Sequence-based anytime control,'' \emph{IEEE
		Trans. Autom. Control}, vol.~58, no.~2, pp. 377--390, 2013.
	
	\bibitem{anytime2}
	\BIBentryALTinterwordspacing
	D.~E. Quevedo, W.-J. Ma, and V.~Gupta, ``Anytime control using input sequences
	with {M}arkovian processor availability,'' \emph{IEEE Trans. Autom. Control},
	vol.~60, no.~2, pp. 515--521, Feb. 2015. [Online]. Available:
	\url{https://arxiv.org/pdf/1405.0751.pdf}
	\BIBentrySTDinterwordspacing
	
	\bibitem{anytime3}
	D.~E. {Quevedo}, V.~{Gupta}, W.~{Ma}, and S.~{Yüksel}, ``Stochastic stability
	of event-triggered anytime control,'' \emph{IEEE Trans. Autom. Control},
	vol.~59, no.~12, pp. 3373--3379, 2014.
	
	\bibitem{anytime4}
	T.~V. {Dang}, K.~{Ling}, and D.~E. {Quevedo}, ``Stability analysis of
	event-triggered anytime control with multiple control laws,'' \emph{IEEE
		Trans. Autom. Control}, vol.~64, no.~1, pp. 420--426, 2019.
	
	\bibitem{ZHAN2015214}
	\BIBentryALTinterwordspacing
	X.-S. Zhan, J.~Wu, T.~Jiang, and X.-W. Jiang, ``Optimal performance of
	networked control systems under the packet dropouts and channel noise,''
	\emph{ISA Transactions}, vol.~58, pp. 214--221, 2015. [Online]. Available:
	\url{https://www.sciencedirect.com/science/article/pii/S0019057815001299}
	\BIBentrySTDinterwordspacing
	
	\bibitem{Parastoo}
	P.~{Sadeghi}, R.~A. {Kennedy}, P.~B. {Rapajic}, and R.~{Shams}, ``Finite-state
	{Markov} modeling of fading channels - a survey of principles and
	applications,'' \emph{IEEE Signal Process. Mag.}, vol.~25, no.~5, pp. 57--80,
	Sep. 2008.
	
	\bibitem{KangJIoT}
	K.~{Huang}, W.~{Liu}, Y.~{Li}, B.~{Vucetic}, and A.~{Savkin}, ``Optimal
	downlink-uplink scheduling of wireless networked control for {Industrial
		IoT},'' \emph{IEEE Internet Things J.}, vol.~7, no.~3, pp. 1756--1772, Mar.
	2020.
	
	\bibitem{Kang2019ICC}
	K.~Huang, W.~Liu, Y.~Li, and B.~Vucetic, ``To retransmit or not: Real-time
	remote estimation in wireless networked control,'' in \emph{Proc. IEEE ICC},
	May 2019, pp. 1--7.
	
	\bibitem{KangTWC}
	K.~Huang, W.~Liu, M.~Shirvanimoghaddam, Y.~Li, and B.~Vucetic, ``Real-time
	remote estimation with hybrid {ARQ} in wireless networked control,''
	\emph{IEEE Trans. Wireless Commun.}, vol.~19, no.~5, pp. 3490--3504, 2020.
	
	\bibitem{tse2005fundamentals}
	D.~Tse and P.~Viswanath, \emph{Fundamentals of wireless communication}.\hskip
	1em plus 0.5em minus 0.4em\relax Cambridge university press, 2005.
	
	\bibitem{Polyanskiy}
	Y.~Polyanskiy, H.~V. Poor, and S.~Verdu, ``Channel coding rate in the finite
	blocklength regime,'' \emph{IEEE Trans. Inf. Theory}, vol.~56, no.~5, pp.
	2307--2359, May 2010.
	
	\bibitem{Goldsmith}
	A.~J. {Goldsmith} and {Soon-Ghee Chua}, ``Variable-rate variable-power mqam for
	fading channels,'' \emph{IEEE Trans. Commun.}, vol.~45, no.~10, pp.
	1218--1230, 1997.
	
	\bibitem{schenato2007foundations}
	L.~Schenato, B.~Sinopoli, M.~Franceschetti, K.~Poolla, and S.~S. Sastry,
	``Foundations of control and estimation over lossy networks,'' \emph{Proc.
		IEEE}, vol.~95, no.~1, pp. 163--187, Jan. 2007.
	
	\bibitem{QUEVEDO20121803}
	D.~E. Quevedo and D.~Nešić, ``Robust stability of packetized predictive
	control of nonlinear systems with disturbances and markovian packet losses,''
	\emph{Automatica}, vol.~48, no.~8, pp. 1803 -- 1811, 2012.
	
	\bibitem{liu2020remote}
	\BIBentryALTinterwordspacing
	W.~Liu, D.~E. Quevedo, Y.~Li, K.~H. Johansson, and B.~Vucetic, ``Remote state
	estimation with smart sensors over markov fading channels,'' \emph{submitted
		to IEEE Trans. Autom. Control}, 2020. [Online]. Available:
	\url{https://arxiv.org/abs/2005.07871}
	\BIBentrySTDinterwordspacing
	
	\bibitem{liu2021remote}
	\BIBentryALTinterwordspacing
	W.~Liu, D.~E. Quevedo, K.~H. Johansson, Y.~Li, and B.~Vucetic, ``Remote state
	estimation of multiple systems over multiple markov fading channels,''
	\emph{submitted to IEEE Trans. Autom. Control}, 2021. [Online]. Available:
	\url{https://arxiv.org/abs/2104.04181}
	\BIBentrySTDinterwordspacing
	
	\bibitem{durrett2019probability}
	R.~Durrett, \emph{Probability: Theory and Examples}.\hskip 1em plus 0.5em minus
	0.4em\relax Cambridge university press, 2019, vol.~49.
	
	\bibitem{myarxiv}
	\BIBentryALTinterwordspacing
	W.~Liu, D.~E. Quevedo, Y.~Li, and B.~Vucetic, ``Anytime control under practical
	communication model,'' \emph{submitted to IEEE Trans. Autom. Control}, 2021.
	[Online]. Available: \url{https://arxiv.org/abs/2012.00962}
	\BIBentrySTDinterwordspacing
	
	\bibitem{Perron}
	H.~Minc, \emph{Nonnegative Matrices}.\hskip 1em plus 0.5em minus 0.4em\relax
	Wiley, 1988.
	
	\bibitem{mymatrix}
	\BIBentryALTinterwordspacing
	\emph{State Transition Matrix of $Z(t)$}, May 2021. [Online]. Available:
	\url{{https://drive.google.com/file/d/12hnTfMwKoM9Rx4eXSlq9N0Sak\_6M8dtb/view?usp=sharing/}}
	\BIBentrySTDinterwordspacing
	
	\bibitem{costa2006discrete}
	O.~L.~V. Costa, M.~D. Fragoso, and R.~P. Marques, \emph{Discrete-time Markov
		jump linear systems}.\hskip 1em plus 0.5em minus 0.4em\relax Springer Science
	\& Business Media, 2006.
	
\end{thebibliography}
\end{document}